\documentclass[aps,pre,reprint,superscriptaddress,longbibliography]{revtex4-2}

\usepackage[utf8]{inputenc}
\usepackage[T1]{fontenc}
\usepackage{hyperref}
\usepackage{url}
\usepackage{booktabs}
\usepackage{amsfonts}
\usepackage{nicefrac}
\usepackage{microtype}
\usepackage{xcolor}
\usepackage{graphicx}
\usepackage{amsmath,amssymb,amsthm}
\usepackage{bm}

\newcommand{\R}{\mathbb{R}}
\newcommand{\E}{\mathbb{E}}
\newcommand{\sign}{\mathrm{sign}}

\newtheorem{proposition}{Proposition}
\newtheorem{definition}{Definition}
\newtheorem{remark}{Remark}

\graphicspath{{figures/}}

\begin{document}

\title{Finite-size scaling of hetero-associative retrieval in continuous-signal-driven Ising spin systems}

\author{Andrea Ladiana}
\affiliation{Dipartimento di Scienze di Base e Applicazioni all'Ingegneria, Sapienza Universit\`a di Roma, Italy}
\affiliation{Istituto Nazionale di Alta Matematica Francesco Severi (INdAM), Roma, Italy}

\begin{abstract}
Real-world physical signals are continuous and high-dimensional, yet the statistical-mechanics machinery of associative memory operates on discrete Ising spins. We bridge this divide through a multilayer Ising framework that couples a geometry-preserving continuous-to-Ising encoder (PCA whitening composed with SimHash random-hyperplane projection) to Kanter--Sompolinsky pseudo-inverse memory couplings, embedded directly into the local-field equations of a tri-layer hetero-associative system. The pseudo-inverse correction renders the equal-weight mixture state thermodynamically unstable, so that thermal fluctuations break the cross-modal symmetry and select a single global winner. We further establish a dynamical duality: parallel (Little) updates are structurally required to ignite the cross-modal signal avalanche from a single cued layer, whereas sequential (Glauber) sweeps resolve symmetric superpositions. The operational storage capacity obeys the Amit--Gutfreund--Sompolinsky finite-size correction $\alpha_c(N)=\alpha_c(\infty)-c\,N^{-1/2}$, extrapolating to an asymptotic operational limit $\alpha_c(\infty)\approx 0.50$ under macroscopic-basin retrieval. Applied to multi-channel sleep polysomnography (PhysioNet Sleep-EDF), the architecture reconstructs the macroscopic sleep state on parietal EEG and EOG axes from a single noisy frontal-EEG cue, demonstrating cross-modal recall in the presence of quenched biological disorder.
\end{abstract}

\maketitle

\section{Introduction}
\label{sec:intro}

Content-addressable memory---the recovery of a complete stored record from a partial or degraded cue---is a cornerstone of neural computation and a paradigmatic problem in the statistical mechanics of disordered systems. The Hopfield model~\cite{hopfield1982neural} demonstrated that binary patterns in $\{-1,+1\}^N$ can be stored as fixed-point attractors of an energy landscape shaped by Hebbian couplings, while Little's earlier stochastic formulation~\cite{little1974existence} cast the same architecture in the language of equilibrium statistical mechanics through parallel probabilistic updates. The storage capacity of these networks was rigorously characterised by Amit, Gutfreund and Sompolinsky~\cite{amit1985storing,amit1987statistical}, who established the threshold $\alpha_c \approx 0.138$ patterns per spin for Hebbian learning. Projection learning rules~\cite{personnaz1986information} and their Kanter--Sompolinsky variant~\cite{kanter1987associative} raised this bound toward $\alpha_c = 1$ by eliminating inter-pattern cross-talk, at the cost of requiring global knowledge of all stored patterns; dreaming-based consolidation procedures~\cite{fachechi2019dreaming,agliari2019dreaming,agliari2024regularization,agliari2024hebbian} provide a complementary route to higher capacity through iterative suppression of spurious attractors. Dense associative memories~\cite{krotov2016dense,demircigil2017model} and the identification of Hopfield dynamics within transformer attention layers~\cite{ramsauer2021hopfield} have recently reignited interest in energy-based retrieval, reframing the learning problem through the lens of statistical mechanics~\cite{bahri2020statistical,zdeborova2016statistical}; sparse distributed representations~\cite{kanerva1988sparse} provide a complementary pathway through high-dimensional binary codes. Across all these advances, however, a single interface problem persists: real-world signals are continuous and high-dimensional, while the retrieval machinery operates on discrete Ising spins.

\emph{Hetero-associative} architectures address a complementary limitation. Where autoassociative models store and recall within a single spin population, hetero-associative networks couple two or more populations, each representing a distinct ``view'' of the same physical observable. Kosko's Bidirectional Associative Memory~\cite{kosko1988bidirectional} first formalised this idea for two layers, showing that stable recall arises from the same energy-descent principle as in Hopfield networks but across two distinct pattern spaces: a cue presented to one layer drives the other toward the matching stored pattern, enabling cross-modal completion. Multi-species spin-glass analyses~\cite{barra2018multi} subsequently provided a rigorous thermodynamic foundation for architectures with an arbitrary number of interacting populations, generalising the classical single-species phase diagram. More recent work has characterised the thermodynamics of bidirectional memories~\cite{barra2023thermodynamics}, established rigorous capacity bounds for generalised hetero-associative architectures~\cite{agliari2025generalized}, and developed both supervised and unsupervised learning protocols~\cite{alessandrelli2025supervised}. The tri-layer system we study, with populations $(\sigma,\tau,\phi)$ coupled pairwise, is the simplest non-trivial extension beyond two-layer BAM and supports three-way cross-modal completion from a single partial cue.

Coupling continuous signals to a binary associative memory raises three concrete obstacles. First, naive binarisation (e.g.\ thresholding) discards the metric structure of the data: two nearby points in feature space may map to distant spin configurations. Classical symbolic representations such as SAX~\cite{lin2003symbolic} address this obstacle for scalar time series, but do not preserve angular neighbourhood geometry in high-dimensional feature space; alternative approaches keep memories in continuous state space~\cite{millidge2022universal} and bypass binarisation entirely, at the cost of leaving the Ising framework and its statistical-mechanics tractability. Second, empirical class-mean archetypes computed in spin space inherit inter-class correlations from the encoding, causing retrieval basins to overlap and merge---a failure mode we term \emph{basin collapse}. Third, temporal signals are highly sensitive to phase shifts: standard sequence flattening maps identical, time-shifted sequences to orthogonal spin vectors, precluding shift-invariant retrieval.

To resolve these obstacles systematically, we develop a modular three-stage framework. The \emph{continuous-to-Ising encoder} (Sec.~\ref{sec:encoding}) composes phase-invariant preprocessing (e.g.\ the Short-Time Fourier Transform), PCA whitening, and random-hyperplane hashing (SimHash~\cite{hadi2015detection,charikar2002similarity}) into a map from continuous vectors to Ising spins that preserves angular neighbourhood structure with sub-Gaussian overlap fluctuations of standard deviation $O(N^{-1/2})$. \emph{Archetype extraction with coupling correction} (Secs.~\ref{sec:archetypes}--\ref{sec:sharpening}) combines majority-vote archetypes with Kanter--Sompolinsky~\cite{kanter1987associative} pseudo-inverse couplings to eliminate inter-pattern cross-talk. \emph{Multilayer retrieval dynamics} (Sec.~\ref{sec:dynamics}) injects the pseudo-inverse correction inside the local field of a generic $L$-layer hetero-associative system, and operates it via stochastic Monte Carlo updates with simulated annealing~\cite{kirkpatrick1983optimization}; we specialise to $L=3$ throughout the numerical experiments.

The investigation that follows maps the thermodynamic phase crossovers of the system, quantifies the finite-size scaling of its storage capacity, and characterises the stochastic structure of spontaneous symmetry breaking during the disentanglement of superposed memory states. Crucially, to demonstrate that these mechanisms are not artefacts of identically distributed synthetic noise, we apply the framework to multi-channel physiological sleep recordings, showing successful cross-modal hetero-associative retrieval in the presence of strong, correlated biological disorder. Where quantitative claims are made, we accompany them with operational definitions, confidence intervals, and explicit baselines; where a behaviour cannot be rigorously established, we report empirical estimates and refrain from overstating their asymptotic content.

\section{The Encoding Framework}
\label{sec:pipeline}

\subsection{Continuous-to-Ising encoding}
\label{sec:encoding}

Let $\bm{x}\in\R^{d\times T}$ be a multivariate sequence with $d$ channels and $T$ time steps. To address the phase sensitivity of temporal signals, $\bm{x}$ may first be mapped to a phase-invariant representation, such as the magnitude spectrum of a Short-Time Fourier Transform (STFT) (detailed in Sec.~\ref{sec:exp6}). We then flatten the (transformed or raw) sequence to a vector $\bm{v}\in\R^D$, with $D=dT$ for the raw flattening; throughout, symbols with index superscripts (e.g.\ $\xi^\mu$) denote pattern-indexed quantities, whereas $\bm{\bar v}$ denotes the training-set mean. The four-step encoding then reads:
\begin{enumerate}
  \item \textbf{Centering:}
    $\bm{v}_c = \bm{v} - \bm{\bar v}$.
  \item \textbf{PCA projection and whitening:}
    $\bm{z} = \Lambda^{-1/2} V \bm{v}_c \in \R^{d_{\mathrm{eff}}}$,
    where $V \in \R^{d_{\mathrm{eff}} \times D}$ contains the top
    eigenvectors and $\Lambda = \mathrm{diag}(\lambda_1, \ldots,
    \lambda_{d_{\mathrm{eff}}})$ their eigenvalues.
  \item \textbf{Random-hyperplane projection:}
    $\bm{h} = H \bm{z} \in \R^N$, where $H \in \R^{N \times d_{\mathrm{eff}}}$
    has i.i.d.\ $\mathcal{N}(0,1)$ entries.
  \item \textbf{Sign binarisation:}
    $\bm{s} = \sign(\bm{h}) \in \{-1,+1\}^N$.
\end{enumerate}

The composite map is
\begin{equation}
  \label{eq:encoding}
  s(\bm{v}) = \sign\!\big(H\Lambda^{-1/2}V(\bm{v}-\bm{\bar v})\big)
  \;\in\; \{-1,+1\}^N.
\end{equation}
When more than one encoder is used (Sec.~\ref{sec:dynamics}) we write $s_\ell(\bm v)$ for the encoder of layer~$\ell$, with layer-specific matrices $(H_\ell,\Lambda_\ell,V_\ell)$.

This pipeline does not propose another heuristic for machine-learning performance; rather, it establishes a formal dictionary mapping classical signal-processing operations to the observables of equilibrium statistical mechanics. PCA whitening is not used merely for dimensionality reduction, but as a \emph{canonical transformation} that maps the highly structured data manifold onto an isotropic space of maximal entropy. SimHash is then deployed not as a fast-search trick, but as a \emph{microcanonical quantisation} that translates the continuous Euclidean inner product directly into a Mattis magnetisation---the natural order parameter driving Glauber dynamics. This dictionary turns a continuous pattern-recognition pipeline into the Hamiltonian of an Ising spin glass, allowing us to predict both successes and failure modes from first principles (e.g.\ the algebraic constraints of sign re-binarisation in Sec.~\ref{sec:sharpening}) rather than hide them behind end-to-end optimisation. For purely unstructured input noise, the whitening operator reduces to the identity, recovering the standard theoretical framework without information loss.

When class labels are available and the signal-to-noise ratio is particularly low (as in the biological data of Sec.~\ref{sec:exp8}), the unsupervised PCA whitening can be replaced by supervised Linear Discriminant Analysis (LDA): the top eigenvectors of $\Sigma_W^{-1}\Sigma_B$ define a whitened subspace that actively filters out state-invariant biological noise before the SimHash binarisation. This substitution leaves the thermodynamic guarantees intact while improving the separation gap~$\Delta q$ (at the cost of requiring labels at encoding time).

\begin{proposition}[SimHash overlap concentration]
\label{prop:simhash}
Let $\bm{z}_1, \bm{z}_2$ be two vectors in the whitened PCA space
with angle $\theta = \arccos(\hat{\bm{z}}_1 \cdot \hat{\bm{z}}_2)$,
and let each row of $H$ be drawn i.i.d.\ from
$\mathcal{N}(0,I_{d_\mathrm{eff}})$.  Then the per-spin overlap
$q_N := N^{-1}\bm{s}_1\cdot\bm{s}_2$ satisfies
\begin{equation}
  \label{eq:simhash}
  \E\!\left[q_N\right] = 1 - \frac{2\theta}{\pi},
\end{equation}
and, by Hoeffding's inequality, concentrates around its mean as a
sub-Gaussian random variable of standard deviation $O(N^{-1/2})$:
for every $t>0$, $\Pr\!\big[|q_N - \E q_N|\ge t\big]\le 2\exp(-Nt^2/2)$.
\end{proposition}

\begin{proof}
Because a Gaussian vector in $\R^{d_\mathrm{eff}}$ has rotationally invariant direction, each normalised row $\bm{h}_i/\|\bm{h}_i\|$ is uniform on the unit sphere~\cite{goemans1995improved}; the sign of $\bm{h}_i\cdot\bm z$ is invariant under positive rescaling, so the $H\sim\mathcal{N}(0,1)^{\otimes}$ construction is statistically equivalent to sampling uniform normals on $S^{d_\mathrm{eff}-1}$. The probability that $\bm{z}_1$ and $\bm{z}_2$ lie on opposite sides of hyperplane $i$ is therefore $\theta/\pi$~\cite{indyk1998approximate,charikar2002similarity,goemans1995improved}, so $\Pr[s_i^{(1)}\ne s_i^{(2)}] = \theta/\pi$ and $\E[s_i^{(1)} s_i^{(2)}] = 1 - 2\theta/\pi$. The $N$ indicators $\mathbf{1}[s_i^{(1)}=s_i^{(2)}]$ are i.i.d.\ and bounded in $[0,1]$; Hoeffding's inequality yields the stated sub-Gaussian tail, and the variance bound $\tfrac{1}{N}\mathrm{Var}[\mathbf{1}_{s_i^{(1)}=s_i^{(2)}}]\le \tfrac{1}{4N}$ gives the $O(N^{-1/2})$ standard deviation.
\end{proof}

\begin{remark}[Overlap magnitude and the separation gap]
\label{rem:overlap_magnitude}
Equation~\eqref{eq:simhash} is exact but operates on angles in the \emph{whitened} PCA space. After whitening, intra-class samples subtend large angles (empirically $\theta\approx 87^\circ$ for our benchmark at $r{=}0.8$), yielding modest absolute overlaps $q\approx 0.04$--$0.07$. The encoding nevertheless creates meaningful thermodynamic structure because the intra-class overlap consistently exceeds the inter-class overlap, producing a positive \emph{separation gap}~$\Delta q$ that seeds the retrieval dynamics. The sub-Gaussian concentration of Proposition~\ref{prop:simhash} keeps this gap stable at standard deviation $O(N^{-1/2})$, with empirical insensitivity to $N$ for $N\geq 32$ (Sec.~\ref{sec:exp1}).
\end{remark}

\subsection{Empirical archetypes in spin space}
\label{sec:archetypes}

For each class $\mu$ we define an empirical archetype \emph{after} encoding,
\begin{equation}
  \label{eq:archetype}
  \xi_i^{\mu} = \sign\!\left(\frac{1}{M_\mu}
  \sum_{a=1}^{M_\mu} s_i^{(\mu,a)}\right),
  \qquad i = 1,\ldots,N,
\end{equation}
where $s^{(\mu,a)}$ is the spin encoding of the $a$-th training example of class~$\mu$. This majority-vote rule selects the binary pattern that maximises total Mattis overlap with the class's data cloud.

The \emph{archetype Gram matrix}
\begin{equation}
  \label{eq:gram}
  G_{\mu\nu}
  = \frac{1}{N}\,\bm{\xi}^\mu \cdot \bm{\xi}^\nu
\end{equation}
has identically unit diagonal. Large off-diagonal entries $|G_{\mu\nu}|$ signal that archetypes $\mu$ and $\nu$ share many spin components, which causes their retrieval basins to merge in the basin-collapse failure mode.

\begin{definition}[Basin Collapse Index]
\label{def:bci}
The basin collapse index is
\begin{equation}
  \mathrm{BCI}
  \;=\; \frac{1}{K(K-1)}
  \sum_{\substack{\mu,\nu=1\\ \mu\neq\nu}}^{K}
  |G_{\mu\nu}|,
\end{equation}
where the summation runs over all $K(K-1)$ \emph{ordered} index pairs $(\mu,\nu)$ with $\mu\neq\nu$; equivalently, $\mathrm{BCI}$ is the arithmetic mean of the absolute off-diagonal entries of $G$. The bound $0 \leq \mathrm{BCI} \leq 1$ encodes the orthogonality structure: $\mathrm{BCI} \to 0$ corresponds to quasi-orthogonal archetypes and well-separated basins, while large values indicate strong cross-talk. The absolute value is essential, as both positive and negative inter-archetype correlations contribute to collapse risk.
\end{definition}

\subsection{Spectral sharpening and pseudo-inverse couplings}
\label{sec:sharpening}

When archetype correlations are significant, the Kanter--Sompolinsky remedy~\cite{kanter1987associative} replaces Hebbian couplings with pseudo-inverse couplings,
\begin{equation}
  \label{eq:pseudo_inv}
  J^{\mathrm{pseudo}}_{ij}
  = \frac{1}{N}\sum_{\mu,\nu} \xi_i^\mu\,
    \big[(G + \gamma I_K)^{-1}\big]_{\mu\nu}\, \xi_j^\nu,
\end{equation}
where $\gamma>0$ is a Tikhonov regulariser. The construction satisfies $J^{\mathrm{pseudo}}\bm{\xi}^\mu \approx \bm{\xi}^\mu$ for each archetype, removing the cross-talk responsible for basin collapse.

An equivalent perspective---\emph{spectral sharpening}---defines the continuous patterns $\widetilde{\Xi} = (G + \gamma I_K)^{-1} \Xi$, so that Hebbian learning on $\widetilde{\Xi}$ recovers
\begin{equation}
  \label{eq:sharpened_hebbian}
  \tilde{J} = \frac{1}{N}\widetilde{\Xi}^T\widetilde{\Xi}
  = \frac{1}{N}\Xi^T(G + \gamma I)^{-2}\Xi.
\end{equation}
The continuous patterns $\widetilde{\Xi}$ are nearly orthogonal in our benchmark ($\mathrm{BCI}\approx 0.001$ vs.\ $0.32$ before sharpening). However, applying $\sign(\cdot)$ to re-binarise them is provably ineffective whenever the underlying patterns are sufficiently anti-correlated:

\paragraph{Sign re-binarisation is a no-op for diagonally-dominant sharpening.}
\label{sec:sign_ineffective}
Let $\Xi\in\{-1,+1\}^{K\times N}$ be $K$ binary archetypes with Gram matrix $G\in\R^{K\times K}$, let $A=(G+\gamma I)^{-1}$ with $\gamma>0$ chosen so that $G+\gamma I\succ 0$, and suppose $A$ is strictly row-diagonally dominant in absolute value, $A_{\mu\mu}>\sum_{\nu\neq\mu}|A_{\mu\nu}|$ for every~$\mu$. Then $\sign(A\,\Xi)=\Xi$ entry-wise. Indeed, fixing $\mu,i$ and writing $\tilde\xi^\mu_i=\sum_\nu A_{\mu\nu}\xi^\nu_i$, the diagonal contribution $A_{\mu\mu}\xi^\mu_i$ has sign $\xi^\mu_i$ and magnitude $A_{\mu\mu}$, while $\big|\sum_{\nu\neq\mu}A_{\mu\nu}\xi^\nu_i\big|\le \sum_{\nu\neq\mu}|A_{\mu\nu}|<A_{\mu\mu}$. Hence the diagonal term dominates and $\sign(\tilde\xi^\mu_i)=\sign(\xi^\mu_i)$.

\paragraph{When the hypothesis holds.}
For the equicorrelated case relevant to our $K{=}3$ benchmark, $G=(1-\rho)I+\rho J$ with $\rho=-0.31$ ($J$ being the all-ones $K\times K$ matrix). A direct calculation gives $A=a I+b J$ with
\begin{equation*}
  \begin{aligned}
     a&=\frac{1}{(1-\rho)+\gamma}, \\
     b&=-\frac{\rho}{\big((1-\rho)+\gamma\big)\big((1-\rho)+\gamma+K\rho\big)},
  \end{aligned}
\end{equation*}
so $A_{\mu\mu}=a+b$ and $A_{\mu\nu}=b$ for $\nu\neq\mu$. Strict row diagonal dominance in absolute value, $|a+b|>(K-1)|b|$, is equivalent (for $\rho<0$, in the regime $G+\gamma I\succ 0$) to the condition $1+(2K{-}3)\rho+\gamma > 0$, i.e.\ $|\rho|<(1+\gamma)/(2K{-}3)$. For $K{=}3,\rho{=}-0.31,\gamma{=}10^{-3}$ this reads $0.31<0.3337$, and direct numerical verification confirms zero bit-flips. For larger $K$ or non-equicorrelated $G$, the bound tightens and the hypothesis typically fails: sharpening followed by sign re-binarisation does flip bits, consistent with Fig.~\ref{fig:merged_bci_recon}(a).

The physical consequence is sharp: for the uniformly anti-correlated $K{=}3$ geometry of our benchmark, sign re-binarisation produces archetypes \emph{identical} to the originals. Decorrelation survives only in the continuous representation $\widetilde{\Xi}$, which cannot be used directly inside binary spin dynamics. We must therefore implement the pseudo-inverse correction \emph{within} the dynamics rather than in the pattern space (Sec.~\ref{sec:dynamics}). Spectral sharpening regains its bite for larger $K$ where archetypes are no longer uniformly anti-correlated (Sec.~\ref{sec:exp2}).

\section{Multilayer Hetero-Associative Dynamics}
\label{sec:dynamics}

We first present the retrieval equations in a generic $L$-layer setting, then specialise to $L{=}3$ (layers $(\sigma,\tau,\phi)$) for the numerical experiments. Unless stated otherwise, symbols $(\sigma,\tau,\phi)$ refer to this tri-layer specialisation. Multi-layer probabilistic architectures with Gibbs-style block updates also appear in the deep-learning literature as Deep Boltzmann Machines~\cite{salakhutdinov2009deep}; the present system differs in that each visible layer encodes a distinct physical channel and the inter-layer couplings are set analytically via pseudo-inverse rather than learned from data.

\subsection{Architecture}
\label{sec:arch}

The memory stores $K$ multimodal patterns $\{\bm{\xi}^{\mu,(\ell)}:\mu=1,\ldots,K,\ \ell=1,\ldots,L\}$ across $L$ binary spin layers,
\begin{equation}
  \bm{s}^{(\ell)} \in \{-1,+1\}^{N_\ell},
  \qquad \ell = 1,\ldots,L.
\end{equation}
Layers are coupled pairwise with symmetric strengths $g_{\ell m}=g_{m\ell}$ for $\ell\neq m$. We specialise to $L{=}3$ with $(\bm{s}^{(1)},\bm{s}^{(2)},\bm{s}^{(3)}) = (\bm{\sigma},\bm{\tau},\bm{\phi})$ and $(\bm{\xi}^{\mu,(1)},\bm{\xi}^{\mu,(2)},\bm{\xi}^{\mu,(3)}) = (\bm{\xi}^\mu,\bm{\eta}^\mu,\bm{\chi}^\mu)$.

Each layer carries its own random-hyperplane matrix $H_\ell$ (with an independent seed), so that the spin representations $\bm{\xi}^\mu,\bm{\eta}^\mu,\bm{\chi}^\mu$ are, with probability one in the Gaussian law on $H_\ell$, mutually distinct. This is what enables the tri-layer system to encode genuinely multimodal data. If all layers shared a single encoder $H$, the input map would collapse to $\bm{s}^{(1)}(\bm v)=\cdots=\bm{s}^{(L)}(\bm v)$: stored archetypes would coincide across layers, all overlap vectors would satisfy $m_1^\mu = \cdots = m_L^\mu$, and the dynamics would synchronise the layers into an effectively autoassociative memory of dimension $N$, destroying the cross-layer completion capability. Encoder independence is a separate property from independence of the raw layer inputs: in the Sleep-EEG experiment (Sec.~\ref{sec:exp8}) the three layers are fed three physically distinct biological channels (Fpz-Cz, Pz-Oz, EOG), so cross-layer information transfer is driven both by encoder diversity and by anatomical signal diversity; we return to this distinction in Sec.~\ref{sec:discussion}.

\subsection{Effective local field with pseudo-inverse couplings}
\label{sec:field}

The local-field equations derived below follow the framework of Ref.~\cite{alessandrelli2025supervised}, where they were first established in the hetero-associative setting with both supervised and unsupervised training protocols.

For a generic target layer $a\in\{1,\ldots,L\}$, the local field is driven by Mattis overlaps from all other layers $b\neq a$. The Mattis overlap (order parameter) with pattern $\mu$ is
\begin{equation}
  m_a^\mu = \frac{1}{N_a}\,\bm{\xi}^{\mu,(a)}\cdot\bm{s}^{(a)}.
\end{equation}

With standard Hebbian couplings, the local field acting on spin $s_i^{(a)}$ is
\begin{equation}
  \label{eq:field_hebbian}
  h_i^{a,\mathrm{Heb}} = \frac{1}{\sqrt{N_a}} \sum_{\mu=1}^K
  W_i^{a,\mu}\, C_a^\mu,
\end{equation}
where
\begin{equation}
  C_a^\mu = \sum_{b \neq a} g_{ab}\, m_b^\mu\,\sqrt{N_b}\,\alpha_{ab},
  \qquad
  \alpha_{ab} = \sqrt{(1+\rho_a)(1+\rho_b)},
\end{equation}
and $W_i^{a,\mu} = \xi_i^{\mu,(a)}/(1+\rho_a)$ are the effective patterns. The per-layer dataset-entropy parameter is $\rho_a = (1-r_a^2)/(M r_a^2)$, where $r_a\in(0,1]$ is the per-layer data-quality parameter (Sec.~\ref{sec:experiments}) and $M$ is the number of training examples per class. The terminology follows Ref.~\cite{alemanno2023supervised}: $\rho_a$ is a strictly monotone proxy for the conditional Shannon entropy $H(\xi_i^{\mu,(a)}\mid \bm\eta_i^{\mu,(a)})$ of the archetype bit given the $M$-block of noisy examples, with $P_e\approx \tfrac12[1-\mathrm{erf}(1/\sqrt{2\rho_a})]$ the Bayes-optimal error and $H=H_2(P_e)$ the corresponding binary entropy. Operationally (Appendix~\ref{app:rho_derivation}) the same $\rho_a$ equals the ratio of per-component noise variance to squared signal in the empirical centroid under our noise model; the two readings are complementary.

The factor $\alpha_{ab}=\sqrt{(1+\rho_a)(1+\rho_b)}$ arises as the product of the per-layer rescaling constants restoring unit normalisation of $W_i^{a,\mu}$, as derived in Appendix~\ref{app:rho_derivation}. The normalisation is chosen so that the pattern-space coupling of Eq.~\eqref{eq:pseudo_inv} carries an overall $1/N$, matching the $O(1)$ Hebbian convention of Refs.~\cite{amit1987statistical,hertz1991introduction}, while the field of Eq.~\eqref{eq:field_hebbian} carries $1/\sqrt{N_a}$ because the summation $\sum_j\xi_j^{\mu,(b)} s_j^{(b)}$ is already absorbed into $m_b^\mu$ (which scales as $O(N_b^{-1/2})$ for random spins by the central limit theorem). The two conventions are therefore consistent: an $O(1)$ field $h_i^{(a)}$ is obtained from an $O(1/N)$ coupling acting on an $O(N^{-1/2})$ overlap.

The Hebbian field preserves degeneracies in the overlap profile: if $m_b^\mu = m_b^\nu$ for all driving layers $b\neq a$, then $C_a^\mu = C_a^\nu$, and the field cannot distinguish patterns $\mu$ and $\nu$. This degeneracy is fatal for mixture disentanglement, as it renders the symmetric mixture state a stable fixed point of the mean-field dynamics.

The Kanter--Sompolinsky prescription removes the degeneracy by \emph{de-mixing} pattern coordinates before they enter the field. Define the drive vector $\bm{C}_a = (C_a^1,\ldots,C_a^K)^\top\in\R^K$ and the Gram matrix $G_a$ of layer-$a$ archetypes. Correlations in $G_a$ couple coordinates in $\bm{C}_a$; the corrected coordinates solve the ridge-stabilised linear system $(G_a + \gamma I_K)\,\widetilde{\bm{C}}_a = \bm{C}_a$, giving $\widetilde{\bm{C}}_a = (G_a + \gamma I_K)^{-1}\bm{C}_a$. The regulariser $\gamma>0$ is required whenever the empirical Gram is ill-conditioned (finite samples, high load). The correction is therefore applied directly in the dynamics:
\begin{equation}
  \label{eq:field_pinv}
  h_i^{a,\mathrm{PI}} = \frac{1}{\sqrt{N_a}} \sum_{\mu=1}^K
  W_i^{a,\mu} \sum_{\nu=1}^K
  \big[(G_a + \gamma I)^{-1}\big]_{\mu\nu}\, C_a^\nu.
\end{equation}
For $L{=}3$ and target layer $a{=}\sigma$, Eq.~\eqref{eq:field_pinv} reduces to the implementation expression with $G_a\equiv G_\xi$. At mean-field level, the substitution replaces $m_{a,t+1}^\mu = f\big(\beta \sum_\nu G_{a,\mu\nu} m_{a,t}^\nu\big)$ with the approximately decoupled $m_{a,t+1}^\mu \approx f(\beta\, m_{a,t}^\mu)$. The reduction is exact in the limit $\gamma\to 0^+$ with $(G+\gamma I)^{-1}G\to I$. For the finite $\gamma=10^{-3}$ used in practice, a residual $O(\gamma)$ coupling intentionally persists, acting as a numerical safeguard against ill-conditioning when inverting the finite-sample Gram of highly correlated real-world data (such as the Sleep-EEG profiles). This safeguard ensures that thermal fluctuations cleanly break the mixture symmetry and produce robust winner-take-all selection across all domains, without floating-point pathology.

\subsection{Stochastic update rules and thermodynamic cooling}
\label{sec:update}

Each Monte Carlo step updates spins through
\begin{equation}
  \label{eq:update}
  s_i^{t+1} = \sign\!\big(\tanh(\beta\, h_i^t) + u_i^t\big),
  \quad u_i^t \sim \mathcal{U}[-1,1],
\end{equation}
with $\beta$ the inverse temperature. The additive uniform noise $u_i^t$ implements a stochastic update with the correct high- and low-temperature limits: at $\beta\to 0$ every update is independent of the field and produces a uniform spin, while at $\beta\to\infty$ the rule reduces to the deterministic $s_i = \sign(h_i^t)$. A fully rigorous ergodicity claim in the sense of detailed balance requires sequential (single-spin) updates. For cross-modal hetero-associative retrieval, however, we deliberately employ \emph{parallel synchronous} (Little) updates~\cite{little1974existence}.

This choice is forced by a physical consideration with no purely algorithmic alternative. If a cue is injected into layer $\sigma$ while $\tau$ and $\phi$ remain in random spin configurations, a sequential sweep $\sigma\to\tau\to\phi$ would immediately destroy the cue: because the layers are bidirectionally coupled, $\sigma$ would sample the overwhelming noise of $\tau$ and $\phi$ before they have had any opportunity to align to the signal. Under parallel dynamics, by contrast, all layers update simultaneously: $\sigma$ reads noise and may degrade temporarily, but \emph{at the same step} $\tau$ and $\phi$ read the pristine cue from $\sigma$ and begin their alignment to the target archetype. At $t=2$, all three layers are partially correlated with the target, triggering a constructive avalanche through which the signal resonantly amplifies across the entire ensemble (Appendix~\ref{app:glauber}). Sequential Glauber sweeps, on the other hand, are precisely the regime needed for mixture disentanglement, where the initial state is already symmetric across layers and the dynamics must break the symmetry rather than propagate a cue. To assist this symmetry breaking, the disentanglement experiments use a linear annealing schedule $\beta(t)\in[\beta_\mathrm{low},\beta_\mathrm{high}]$: low initial temperature permits thermal exploration across competing basins, while gradual cooling crystallises the emerging winner.

\section{Numerical Experiments}
\label{sec:experiments}

The eight experiments below trace a deliberate arc through the framework, each one designed to test the guarantees established by the previous one. We first verify the geometric fidelity of the encoder (Exp.~1) so that intra-class samples cluster in spin space; we then quantify how the resulting archetypes inherit cross-talk and how spectral sharpening attempts to suppress it (Exp.~2). With these encoding guarantees in hand, we test the central function of the architecture---cross-layer pattern completion from a partial cue (Exp.~3)---and the more delicate problem of breaking a symmetric superposition of two stored patterns (Exp.~4). We then probe the thermodynamic phase structure (Exp.~5), the role of phase-sensitive preprocessing in shift-invariant retrieval (Exp.~6), and the finite-size scaling of the storage capacity (Exp.~7), before delivering the architecture to genuinely disordered biological data in the Sleep-EEG benchmark (Exp.~8). Each step exposes a measurable diagnostic; later experiments use, rather than re-establish, what earlier ones have proven.

\subsection{Definition of the retrieval protocols}
\label{sec:task_definitions}

Three initial-value problems probe distinct dynamical capabilities of the multi-layer memory:

\begin{enumerate}
  \item \textbf{Pattern reconstruction (cross-layer completion).} A protocol probing the existence and reachability of the attractor basins. Layer $\sigma$ is initialised with a (possibly noisy) copy of an archetype~$\mu$; layers $\tau$ and $\phi$ are initialised to random spin noise. The system must reconstruct the missing macroscopic state across the uninitialised layers. Because the initial cue explicitly breaks the pattern symmetry, standard Hebbian couplings suffice.

  \item \textbf{Hard mixture disentanglement.} A protocol probing spontaneous symmetry breaking. All layers are simultaneously initialised in a symmetric superposition of two archetypes, $\bm{s}^{(a)}_{t=0} = \sign(w_1 \bm{\xi}^{\mu_1,(a)} + w_2 \bm{\xi}^{\mu_2,(a)})$. The dynamics must spontaneously break the symmetry and select a single, globally consistent winner. This is structurally demanding: the mixture state must be \emph{actively destabilised} through pseudo-inverse couplings, and sequential thermal annealing is required to orchestrate the selection.

  \item \textbf{Easy mixture disentanglement.} A baseline protocol that isolates basin stability from symmetry breaking. All layers receive a partial, highly corrupted cue drawn from the \emph{same} archetype~$\mu$. Because no cross-pattern competition is present, each layer independently recognises its own cue. This control verifies whether retrieval failure stems from an intrinsically unstable basin or specifically from the inability to break a symmetric mixture.
\end{enumerate}

Experiments~1--7 share a common \textbf{synthetic 3D-curves benchmark}: $K{=}3$ archetype curves in $\R^3$ (a helix, a Lissajous figure, and a polynomial S-curve), each sampled at $T{=}100$ time steps and perturbed with a controlled noise model. At each time step, with probability~$r$ the observed point equals the archetype plus isotropic Gaussian noise of standard deviation $\sigma_\mathrm{near}=0.05$; with probability $1-r$ it is replaced by an outlier drawn uniformly on the bounding box $[-A,A]^3$ of the three archetype curves, $A=1.2$ (chosen so that the outlier support strictly contains the data support). Default data quality is $r{=}0.8$ with $M{=}200$ training samples per class. The encoding uses $d_\mathrm{PCA}{=}32$ principal components and $N{=}512$ Ising spins, with Tikhonov regularisation $\gamma{=}10^{-3}$. Unless otherwise stated, all inter-layer pairwise couplings are unit strength ($g_{\ell m}{=}1$ for $\ell\neq m$; in the tri-layer specialisation $g_{\sigma\tau}=g_{\sigma\phi}=g_{\tau\phi}=1$) and dynamics run for $40$ MC steps. Experiment~8 then applies the complete framework to real-world multivariate electrophysiological data (PhysioNet Sleep-EDF), testing whether the mechanisms validated in controlled conditions survive the transition to genuine biological data. A full parameter table is provided in Appendix~\ref{app:params}.

\subsection{Experiment 1: Concentration of overlaps in hypercubic space}
\label{sec:exp1}

The first question for any encoder is whether it faithfully preserves neighbourhood structure: do samples drawn from the same class achieve a systematically higher spin overlap than samples drawn from different classes, and does this separation remain stable as the spin dimension~$N$ varies? Proposition~\ref{prop:simhash} guarantees an answer in expectation; here we verify it empirically and locate the finite-size regime within which the sub-Gaussian concentration becomes effective.

We measure the intra-class overlap $q_\mathrm{intra} = \frac{1}{N}\langle \bm{s}^{(\mu,a)} \cdot \bm{s}^{(\mu,b)} \rangle_{a\neq b}$, the inter-class overlap $q_\mathrm{inter}$, and the separation gap $\Delta q = q_\mathrm{intra} - q_\mathrm{inter}$ as functions of data quality~$r$ and spin dimension $N\in\{8,16,32,64,128,256,512\}$.

Figure~\ref{fig:encoding}(a) shows that $q_\mathrm{intra}$ grows monotonically with~$r$ and tracks the empirical SimHash prediction computed from measured cosine similarities in the whitened space. As anticipated by Proposition~\ref{prop:simhash}, the finite-size regime is plainly visible at $N\in\{8,16\}$, but the curves collapse tightly for $N\geq 32$, signalling the onset of the $O(N^{-1/2})$ concentration. The implication is methodologically important: encoding quality is determined by the angular geometry of the data cloud, not by the spin dimension itself---a feature that decouples the encoder from the downstream memory size. Panel~(b) shows $\Delta q$ growing monotonically with~$r$, while panel~(c) reveals that the median intra-class angle in whitened space is $\approx 87^\circ$: PCA whitening spreads the data nearly isotropically, which is why absolute overlaps are modest. The positive $\Delta q$ that survives this isotropisation is the signal that seeds all downstream retrieval.

\begin{figure}[t]
\centering
\includegraphics[width=\linewidth]{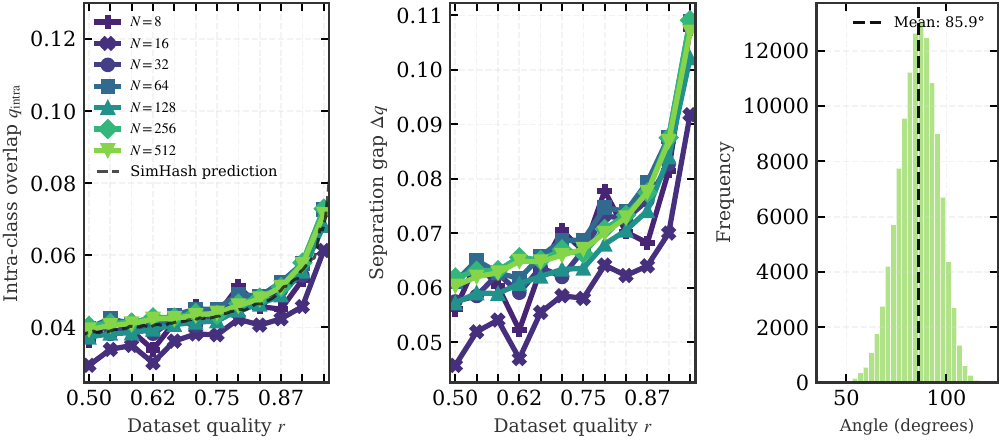}
\caption{\textbf{Concentration of spin overlaps.}
Panel~(a): intra-class overlap $q_\mathrm{intra} = N^{-1}\langle \bm{s}^a \cdot \bm{s}^b \rangle_{a\neq b}$ versus data quality~$r$ for seven spin dimensions $N\in\{8,\ldots,512\}$, alongside the empirical SimHash prediction (dashed black) derived from measured cosine similarities in the whitened PCA space. The visible $N\in\{8,16\}$ deviations and the collapse from $N\geq 32$ jointly confirm the $O(N^{-1/2})$ concentration of Proposition~\ref{prop:simhash}. Panel~(b): the separation gap $\Delta q$ grows monotonically with~$r$, indicating that class structure is faithfully encoded. Panel~(c): PCA whitening spreads intra-class samples nearly isotropically ($r{=}0.8$, median angle $\approx 87^\circ$), yielding modest but persistently positive overlaps. $K{=}3$ classes, $M{=}200$ samples per class.}
\label{fig:encoding}
\end{figure}

\subsection{Experiment 2: Archetype quality and spectral sharpening}
\label{sec:exp2}

Once the encoder is validated, the next question is whether the archetypes it produces are clean enough to serve as memory anchors. We quantify inter-archetype correlations and ask whether spectral sharpening eliminates them, extending the analysis beyond $K{=}3$ to probe scaling.

For the $K{=}3$ benchmark, the empirical Gram matrix has uniformly negative off-diagonal entries $G_{\mu\nu}\approx -0.31$ (BCI~$\approx 0.32$). The sign is a structural consequence of the near-equidistribution of the three curves in whitened space: spins that agree within one class systematically disagree across classes.

Table~\ref{tab:gram} summarises the Gram statistics. As predicted by the diagonal-dominance result of Sec.~\ref{sec:sign_ineffective}, sharpening followed by sign re-binarisation $\sign\big((G+\gamma I)^{-1}\Xi\big)$ leaves \emph{both} BCI and mean off-diagonal entry unchanged to three significant figures---direct confirmation of the no-op. Genuine decorrelation (BCI~$\approx 0.001$) arises only in the continuous sharpened patterns $\widetilde{\Xi}$, which cannot be used in binary spin dynamics. The pseudo-inverse correction must therefore live \emph{inside} the dynamics, as in Eq.~\eqref{eq:field_pinv}.

Figure~\ref{fig:merged_bci_recon}(a) extends the analysis to $K\in\{3,5,8,12,16,20,24,28,32\}$ on a parametric curve family. Continuous sharpening drives the BCI to near zero across the entire range, confirming that the pseudo-inverse prescription scales gracefully; because $J^+$ is recomputed from each pattern set, the fixed-point residual is near-invariant by construction.

\begin{table}[t]
\centering
\caption{\textbf{Effect of spectral sharpening on archetype correlations ($K{=}3$).} Rows correspond to raw majority-vote archetypes, the same archetypes after sharpening followed by sign re-binarisation, and the continuous sharpened patterns. Re-binarisation leaves both quantities unchanged to three significant figures, directly confirming the algebraic no-op predicted for diagonally-dominant matrices; genuine decorrelation is achieved only in the continuous representation $\widetilde{\Xi}$. $N{=}512$, $M{=}200$, $r{=}0.8$, $\gamma{=}10^{-3}$.}
\label{tab:gram}
\begin{tabular}{lcc}
\toprule
Pattern representation & $\bar{G}_{\mu\neq\nu}$ & BCI \\
\midrule
Raw archetypes $\Xi$ & $-0.310$ & $0.320$ \\
Sharpened + re-binarised $\sign\big((G{+}\gamma I)^{-1}\Xi\big)$
  & $-0.310$ & $0.320$ \\
Continuous sharpened $\widetilde{\Xi} = (G{+}\gamma I)^{-1}\Xi$
  & $\phantom{-}0.001$ & $0.001$ \\
\bottomrule
\end{tabular}
\end{table}

\subsection{Experiment 3: Pattern reconstruction and basin structure}
\label{sec:exp3}

With encoding fidelity and archetype quality established, we turn to the central function of the tri-layer architecture: cross-layer pattern completion. Following the \emph{pattern reconstruction} protocol (Sec.~\ref{sec:task_definitions}), layer $\sigma$ is initialised with a (possibly noisy) copy of archetype~$\mu$, while the other two layers begin from random spin configurations. We use Hebbian couplings (no pseudo-inverse correction), as the cued initial state already breaks the pattern symmetry and winner-take-all selection is not required.

Figure~\ref{fig:recon_traj} reports the magnetisation dynamics for pattern $\mu{=}0$ (helix) averaged over $50$ independent runs, starting from a clean cue. The cued layer~$\sigma$ maintains its initial alignment ($m\approx 1$), while $\tau$ and $\phi$ converge from random overlap ($m\approx 0$) to the correct pattern within $5$--$10$ MC steps (top rows). At low~$\beta$ (bottom row), convergence is dominated by thermal fluctuations, illustrating the breakdown of retrieval as one approaches the paramagnetic regime that Sec.~\ref{sec:exp5} will localise quantitatively.

\begin{figure}[t]
\centering
\includegraphics[width=\linewidth]{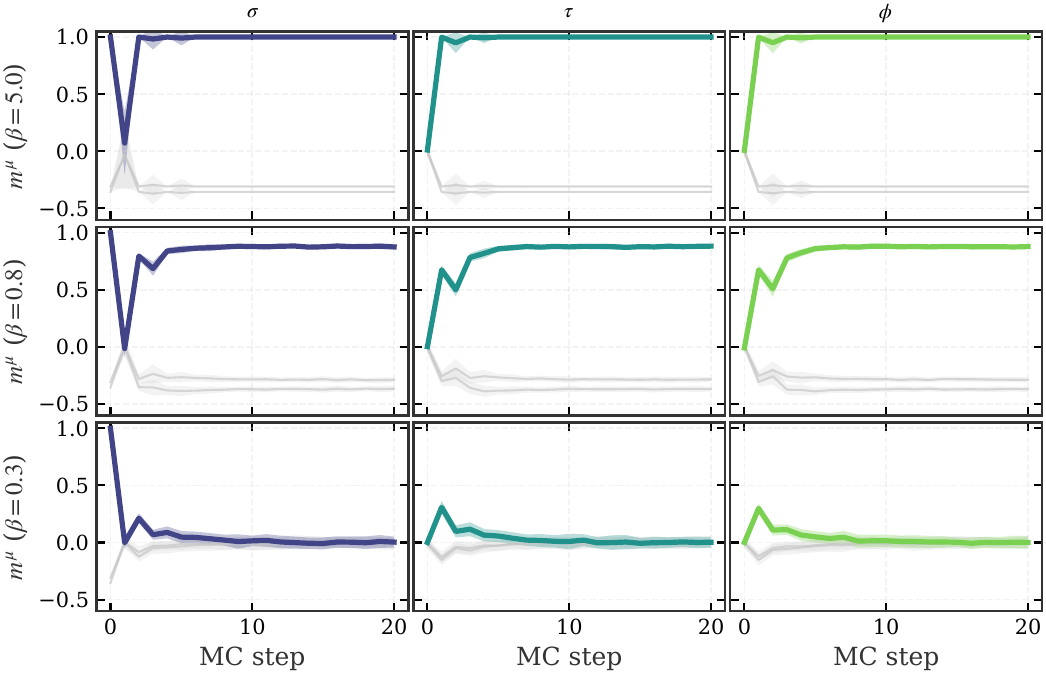}
\caption{\textbf{Cross-layer pattern completion from a clean cue.}
Mattis overlaps $m^\mu = N^{-1}\bm{\xi}^\mu \cdot \bm{s}$ for all $K{=}3$ patterns over MC steps. The cued layer~$\sigma$ (left column) is initialised with the helix archetype ($\mu{=}0$); the uncued layers $\tau$ and $\phi$ (middle and right columns) start from random spin configurations. Rows correspond to inverse temperatures $\beta \in \{5.0,0.8,0.3\}$. Solid lines and shaded bands are the mean and $\pm 1$ s.d.\ over $50$ runs for the target pattern; grey traces show overlaps with the non-target patterns. Hebbian couplings; $N{=}512$.}
\label{fig:recon_traj}
\end{figure}

The more demanding test is cue degradation. We vary the corruption level $p_\mathrm{flip}\in[0,0.5]$ and measure reconstruction success over $K\times 20=60$ trials per noise level. As shown by the graceful degradation in Fig.~\ref{fig:merged_bci_recon}(b), the basins absorb substantial spin-flip noise: $\geq 98\%$ success up to $p_\mathrm{flip}=0.15$, $95\%$ at $p_\mathrm{flip}=0.3$, $77\%$ at $p_\mathrm{flip}=0.4$, collapsing to $13\%$ at the random-cue limit $p_\mathrm{flip}=0.5$. The shape of this curve is the operational fingerprint of macroscopic basin volume, a property we shall need explicitly when interpreting the storage-capacity results of Sec.~\ref{sec:exp7}.

\begin{figure}[t]
\centering
\includegraphics[width=\linewidth]{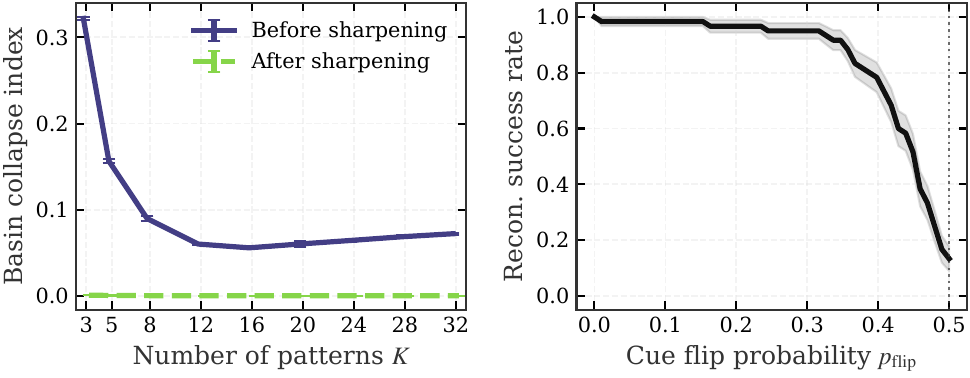}
\caption{\textbf{Sharpening effectiveness and pattern reconstruction robustness.}
Panel~(a): BCI as the number of stored patterns grows from $K{=}3$ to $32$; spectral sharpening (green dashed) reduces BCI by two orders of magnitude relative to the raw archetypes (blue solid) across the entire range ($N{=}512$, $r{=}0.8$). Panel~(b): reconstruction success rate under cue degradation, remaining above $95\%$ up to $30\%$ spin-flip corruption (vertical dotted: random-cue limit). Error bars: $\pm 1$ s.e.\ ($K{=}3$, $20$ seeds).}
\label{fig:merged_bci_recon}
\end{figure}

To confirm that spin-space retrieval translates into physically meaningful continuous reconstructions, Fig.~\ref{fig:recon_3d} projects the retrieved spin configurations back onto the original $\R^3$ space, juxtaposing ground-truth archetypes with reconstructions from cues at increasing corruption levels.

\begin{figure}[t]
\centering
\includegraphics[width=\linewidth]{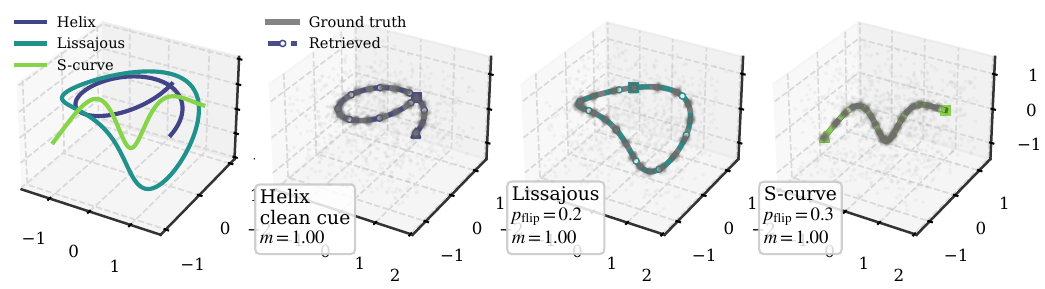}
\caption{\textbf{Retrieval projected back to continuous 3D space.}
Panel~(a): ground-truth archetype curves. Panels~(b--d): retrieved patterns overlaid on the ground truth for each archetype under increasing cue corruption. Semi-transparent point clouds show sampled data for the cued class. Ground truth is rendered as a neutral solid backbone, retrieval as a coloured dashed trajectory with markers; start/end symbols distinguish coincident paths. In-panel annotations report cue condition and final Mattis overlap~$m$. $\beta{=}5$, $40$ MC steps, Hebbian couplings.}
\label{fig:recon_3d}
\end{figure}

\subsection{Experiment 4: Spontaneous symmetry breaking in mixture states}
\label{sec:exp4}

Pattern reconstruction establishes that the attractor basins exist and are reachable from corrupted cues. A structurally harder question remains: can the dynamics separate a binary superposition of two stored patterns, consistently selecting a single winner across all three layers? This is the central diagnostic for the pseudo-inverse correction, which was constructed precisely to destabilise the symmetric mixture fixed point.

Following the \emph{hard disentanglement} protocol (Sec.~\ref{sec:task_definitions}), all layers are initialised with the superposition $\bm{s}^0 = \sign(w_1 \bm{\xi}^{\mu_1} + w_2 \bm{\xi}^{\mu_2})$, and the dynamics must spontaneously break the initial symmetry. The corresponding \emph{easy} control, in which all layers receive a cue drawn from the same archetype, is reported in Appendix~\ref{app:easy_disent} and serves only to confirm that the basins are stable in isolation.

Three ingredients beyond basic pattern reconstruction are now required: \emph{(i)}~independent encoders per layer, producing distinct spin representations that block the degenerate fixed point $\sigma=\tau=\phi$; \emph{(ii)}~pseudo-inverse couplings, which decouple the mean-field dynamics and render the mixture state unstable (Sec.~\ref{sec:field}); and \emph{(iii)}~simulated annealing with sequential Glauber-style sweeps, $\beta:0.5\to 5.0$ over $40$ steps, to prevent the symmetric locking endemic to parallel updates on symmetric initial conditions.

Figure~\ref{fig:disent_traj} illustrates the two regimes through representative equal-weight ($w_1=w_2=1$) trajectories. In the successful case (top row), thermal fluctuations break the symmetry around MC step~$5$, and all layers converge to one component with $m\approx 1$ and a margin~$\sim 1.3$. In the failure case (bottom row), both components remain locked at the spurious mixture fixed point $m^0\approx m^1\approx 0.35$. The success criterion is the event that $m_\mathrm{winner}>0.7$ on all three layers simultaneously. Over $500$ independent seeds the observed success rate is $58.9\%$, with Wilson $95\%$ confidence interval $[54.4\%,63.0\%]$. Because this statistically conservative lower bound is strictly above the chance baseline of $50\%$, the experiment confirms that the initial symmetry is broken by a genuine collective statistical-mechanical phenomenon, not by stochastic accident.

\begin{figure}[t]
\centering
\includegraphics[width=\linewidth]{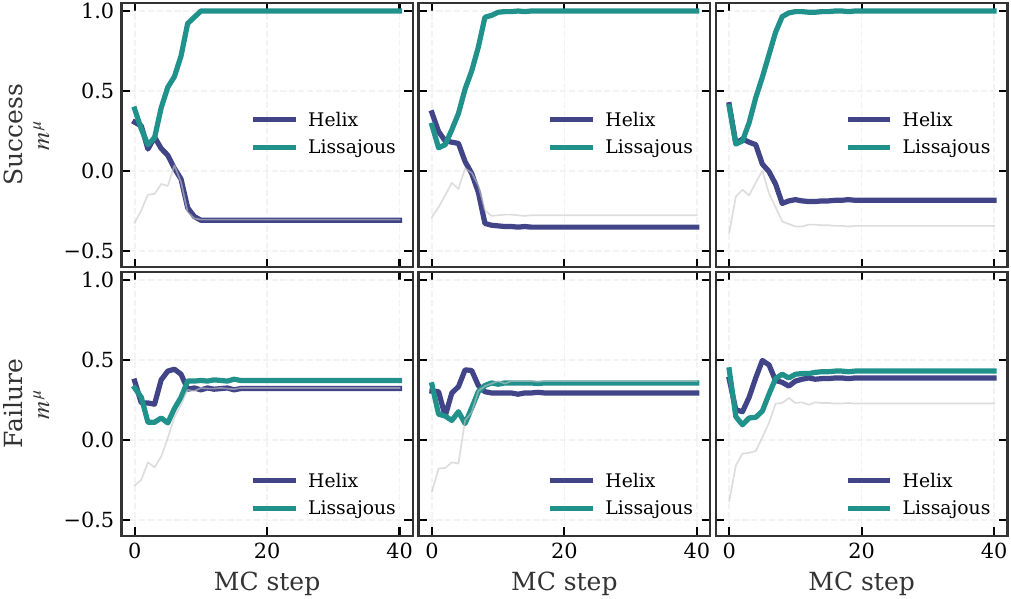}
\caption{\textbf{Mixture disentanglement: spontaneous symmetry breaking.}
Mattis overlaps $m^\mu$ for a $50/50$ mixture of helix and Lissajous patterns, with pseudo-inverse couplings, simulated annealing ($\beta:0.5\to 5.0$), and sequential updates; columns correspond to the three layers $(\sigma,\tau,\phi)$. Top row: thermal fluctuations break the initial symmetry and one component reaches $m\approx 1$ consistently across layers. Bottom row: both components remain trapped at the spurious mixture fixed point. Overall equal-weight success rate: $58.9\%$ over $n=500$ seeds (Wilson $95\%$ CI $[54.4\%,63.0\%]$).}
\label{fig:disent_traj}
\end{figure}

The relative weight of the two components controls the sharpness of the selection. Sweeping the weight ratio $w_1/w_2$ continuously from $1/2$ to $2$ (Fig.~\ref{fig:disent_bias}), the transition is striking: even a slight asymmetry biases the network heavily towards the dominant pattern, and for $w_1/w_2\geq 2$ the selection becomes deterministic ($100\%$ success, one-sided Clopper--Pearson lower bound $0.929$). At exactly equal weights, the symmetry-breaking mechanism produces an even split between the two winners across successful trials, while a substantial fraction of trials remains stuck at the mixture fixed point---the stochastic counterpart of the deterministic asymmetric selection.

A separate diagnostic is cross-layer consistency: the fraction of trials in which all three layers simultaneously converge to the same component with $m>0.7$. As observed previously in discrete asymmetric scenarios, the synchronisation is robust: whenever any single layer succeeds in breaking the symmetry, the others agree on the chosen state. This near-equality between per-layer success rate and cross-layer consistency is the operational signature of the tri-layer coupling actually transmitting macroscopic information between the spin populations.

\begin{figure}[t]
\centering
\includegraphics[width=\linewidth]{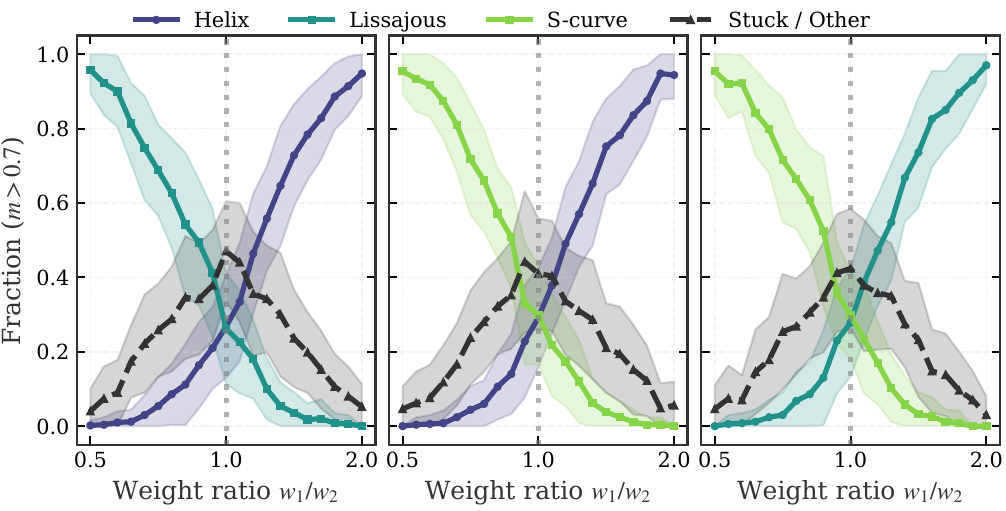}
\caption{\textbf{Continuous selection bias under varying mixture weights.}
Fraction of trials in which each component is selected in layer~$\sigma$ as a function of the weight ratio $w_1/w_2$ (logarithmic grid from $1/2$ to $2$), for all three archetype pairs. Coloured traces: recovery rate of each component (winning criterion $m>0.7$); grey trace: fraction of trials stuck at the spurious mixture fixed point. Solid lines are means over $50$ seeds per condition; shaded bands $\pm 1$ s.d.\ Sequential annealing with pseudo-inverse couplings.}
\label{fig:disent_bias}
\end{figure}

The failure of standard Hebbian couplings in this setting admits a transparent analytical explanation. The mixture initial state has $m^0\approx m^1\approx 0.35$. With Hebbian couplings, the effective field $h_\mathrm{eff}^\mu = \sum_\nu G_{\mu\nu} m^\nu$ preserves the symmetry $m^0=m^1$ exactly (since $G_{0\nu} m^\nu = G_{1\nu} m^\nu$ for equal magnetisations and symmetric Gram rows). The pseudo-inverse couplings replace this with $h_\mathrm{eff}^\mu \approx m^\mu$, so that any thermal fluctuation $\delta m = m^0 - m^1\neq 0$ is amplified by the nonlinearity. This amplification is visible directly in Fig.~\ref{fig:disent_traj}: in the successful top row, a small early imbalance grows rapidly into a macroscopic winner, whereas in the failed bottom row the imbalance never escapes the noise floor and the trajectory remains pinned to the mixture fixed point.

\subsection{Experiment 5: Temperature--overlap phase crossover}
\label{sec:exp5}

The preceding experiments operated within a fixed temperature schedule chosen heuristically. To assess whether that choice is critical---and to expose any dynamical pathologies at extreme inverse temperature---we now sweep $\beta$ systematically across two orders of magnitude, locating the retrieval--paramagnetic boundary (a finite-size \emph{crossover}, not a thermodynamic phase transition, since~$N$ is bounded).

Figure~\ref{fig:beta} reports the mean winner overlap and margin averaged over $30$ independent seeds and all three archetypes, revealing four distinct regimes: a paramagnetic phase at $\beta\lesssim 0.4$, where thermal noise overwhelms the local field; a finite-size retrieval crossover at $\beta_c\approx 0.5$, where the overlap jumps sharply from $\sim 0$ to $\sim 0.7$; a retrieval plateau spanning nearly an order of magnitude ($1\lesssim\beta\lesssim 7$), with near-perfect overlaps $m\approx 1$; and a high-$\beta$ dynamical instability at $\beta\gtrsim 7$, where the overlap drops to $\sim 0.84$ at $\beta=15$.

The crossover at $\beta_c\approx 0.5$ is consistent with the first-order-like jump in the order parameter expected for pseudo-inverse couplings, while the breadth of the retrieval plateau is practically valuable: the system is insensitive to the precise temperature setting within an entire decade, so the annealing schedule is not a fine-tuned hyperparameter. The high-$\beta$ degradation is the more revealing feature. It is not a failure of the memory but a known dynamical artefact of parallel (Little) updates: as $\beta\to\infty$, $\tanh(\beta h)\to\sign(h)$ and the stochastic noise vanishes, producing nearly deterministic simultaneous updates that induce period-$2$ limit cycles~\cite{little1974existence}. Sequential Glauber updates eliminate this artefact entirely (Appendix~\ref{app:glauber}), but parallel updates are structurally required for the cross-modal avalanche, so we accept this high-$\beta$ constraint as a physical cost of cross-layer cue propagation.

\begin{figure}[t]
\centering
\includegraphics[width=\linewidth]{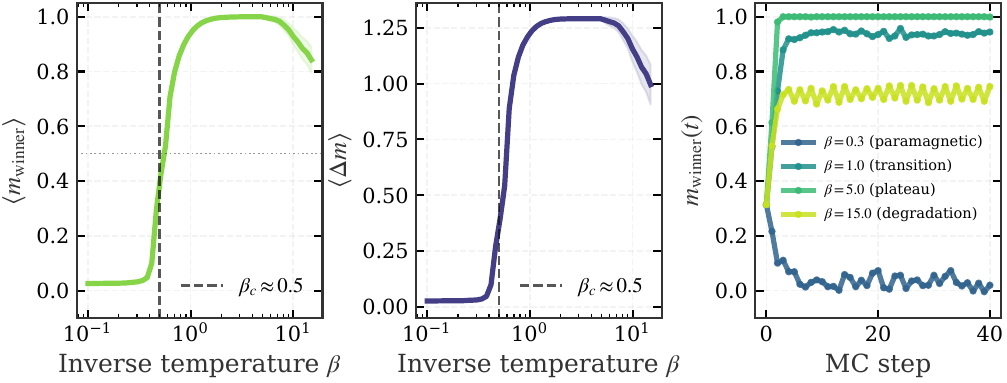}
\caption{\textbf{Temperature--overlap phase diagram under parallel (Little) dynamics.}
Panel~(a): mean winner overlap $\langle m_\mathrm{winner} \rangle$ reveals a finite-size retrieval crossover at $\beta_c\approx 0.5$ (dashed red) and a high-$\beta$ degradation for $\beta\gtrsim 7$. Panel~(b): the margin $\langle\Delta m\rangle = m_\mathrm{winner} - m_\mathrm{second}$ tracks this behaviour (shaded bands: $\pm 2\sigma$ over $90$ trials). Panel~(c): representative $m_\mathrm{winner}(t)$ trajectories illustrate the four regimes (paramagnetic, transitional, plateau, oscillatory) at selected~$\beta$. Clean cue on one layer; Hebbian couplings; $40$ MC steps.}
\label{fig:beta}
\end{figure}

\subsection{Experiment 6: Time-shift disentanglement with STFT}
\label{sec:exp6}

Naive sequence flattening treats raw time-domain samples as an unordered feature vector, conflating identical time-shifted copies of the same underlying curve. Two sequences with identical Euclidean geometry but distinct temporal phase are therefore mapped to indistinguishable SimHash codes, and the memory cannot tell them apart. A natural physical remedy is to replace the flattened representation with a Short-Time Fourier Transform (STFT), whose magnitude spectrum is sensitive to local phase structure. This experiment tests whether the resulting spin-space separation is sufficient for mixture disentanglement of phase-shifted patterns.

Two datasets are generated from the same $3$-D helix: an unshifted version ($t_0=0$) and a $0.25$-period-delayed version ($t_0=0.25$). Both share identical Euclidean geometry and differ only in temporal phase. The training set contains $M=200$ noisy realisations per class ($T=100$ time steps, $r_\mathrm{gen}=0.8$, $\sigma_\mathrm{near}=0.05$). Each multivariate sequence is first transformed by a channel-wise STFT and then passed to the encoder ($N_\mathrm{PCA}=32$, $N=512$).

As shown in Fig.~\ref{fig:stft_disentangle}, starting from the equal-weight mixture of the two patterns and applying simulated annealing ($\beta\in[0.5,5.0]$, $40$ MC steps, parallel Little updates), all three layers converge cleanly to the unshifted pattern: at the final step $m^0_\sigma = m^0_\tau = m^0_\phi = 1.000$, while the competing overlaps settle at $m^1_\sigma=-0.688$, $m^1_\tau=-0.754$, $m^1_\phi=-0.684$. The STFT representation supplies the phase-sensitive discriminability needed to resolve a superposition that naive flattening renders indistinguishable. The mechanism is generic: any phase-sensitive, magnitude-respecting representation (wavelet scattering, mel-frequency cepstrum) would play the same role, which underscores that domain expertise can be injected entirely at the preprocessing stage without altering anything downstream.

\begin{figure}[t]
\centering
\includegraphics[width=\linewidth]{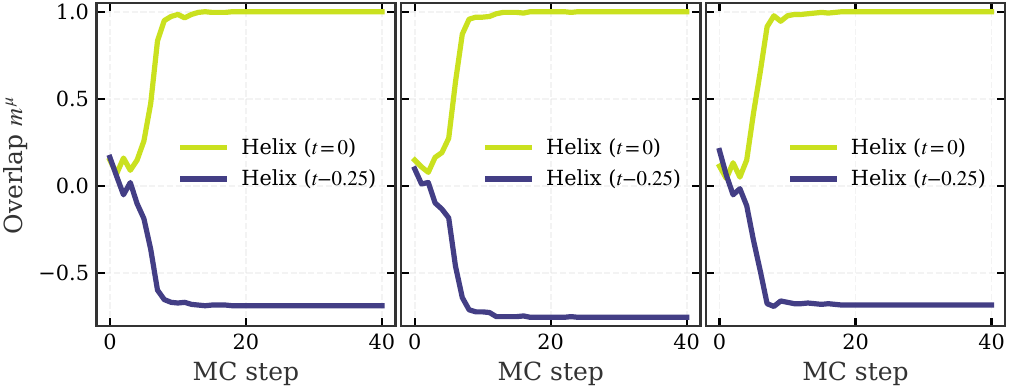}
\caption{\textbf{STFT enables time-shift disentanglement.}
Mattis overlaps $m^\mu$ on the three spin layers $(\sigma,\tau,\phi)$ during simulated annealing ($40$ MC steps, $\beta\in[0.5,5.0]$), starting from an equal-weight superposition of an unshifted helix (red) and its $0.25$-period-delayed counterpart (blue). All layers converge to the unshifted pattern ($m^0\to 1$) with the competitor suppressed ($m^1\to -0.7$), demonstrating phase-invariant separation enabled by the STFT preprocessing.}
\label{fig:stft_disentangle}
\end{figure}

\subsection{Experiment 7: Finite-size storage capacity}
\label{sec:exp7}

Every associative memory has a finite storage capacity, and the scaling of that capacity with the spin dimension governs whether the framework is viable at any non-trivial task complexity. We therefore ask how the number of reliably retrievable patterns grows with~$N$, seeking an empirical estimate of the critical load $\alpha_c=K_c/N$ and a comparison against the theoretical benchmarks of the Hopfield--Kanter--Sompolinsky hierarchy.

For each $N\in\{64,128,256,512,1024\}$ we vary the number of stored patterns~$K$ and measure the retrieval success rate (fraction of random test cues, at $30\%$ spin-flip noise, that converge to the correct archetype with overlap $>0.5$). We adopt the operational definition: $\alpha_c(N;p_\mathrm{flip},\tau)$ is the largest $\alpha=K/N$ at which the success rate remains $\geq \tau$, estimated by linear interpolation of the success-rate curve. Throughout this experiment we fix $p_\mathrm{flip}=0.3$ and $\tau=0.9$; sensitivity to these choices is reported in Appendix~\ref{app:capacity_robustness}. Archetypes are drawn either from the empirical majority-vote procedure on synthetic helix data or as random $\pm 1$ (Rademacher) vectors, the latter serving as a baseline unconstrained by data geometry rather than as a match to the K--S $\alpha_c\to 1$ bound.

As shown by the success-rate curves of Fig.~\ref{fig:capacity}(a), the finite-$N$ capacities exhibit a clean monotonic improvement with system size:
\begin{align}
  \alpha_c(64)   &\approx 0.275, \quad
  \alpha_c(128)  \approx 0.331, \quad
  \alpha_c(256)  \approx 0.365, \nonumber\\
  \alpha_c(512)  &\approx 0.415, \quad
  \alpha_c(1024) \approx 0.456,
  \label{eq:alpha_c}
\end{align}
averaging to $\bar\alpha_c\approx 0.368$. The random-pattern baseline at $N=512$ yields $\alpha_c^{\mathrm{rand}}\approx 0.418$, quantifying the capacity tax imposed by the correlation structure of the empirical archetypes relative to uncorrelated random vectors. Three caveats must be flagged. First, each point in Eq.~\eqref{eq:alpha_c} is a single operating-point estimate: with $20$ seeds per $(N,\alpha)$ cell a binomial $95\%$ CI on a success rate of $0.9$ is approximately $[0.69,0.98]$, which propagates through the interpolation into $\alpha_c$ uncertainties of order $\pm 0.02$--$0.04$ (Appendix~\ref{app:capacity_robustness}). Second, the five points establish a positive finite-size scaling trend, and an empirical FSS analysis lets us identify its form: fitting the observed capacities to the theoretically motivated $\alpha_c(N) \approx \alpha_c(\infty) - c\,N^{-\nu}$ with the standard Amit--Gutfreund--Sompolinsky exponent $\nu=1/2$, the fit is excellent ($R^2\approx 0.97$) and projects to $\alpha_c(\infty)\approx 0.50$ (Fig.~\ref{fig:capacity}b). Third---and this is the load-bearing physical point---this asymptotic operational capacity $\alpha_c(\infty)\approx 0.50$ is not a failure to reach the theoretical Kanter--Sompolinsky limit $\alpha_c\to 1$, but a direct expression of a thermodynamic trade-off. The K--S bound defines the boundary of \emph{local marginal stability}, where patterns persist as local minima but their basins of attraction shrink to zero. Our operational definition demands a \emph{macroscopic basin of attraction}, able to absorb $30\%$ spin-flip corruption; providing such basins necessarily forces the critical capacity below the local-stability bound. Attaining $\alpha_c(\infty)\approx 0.50$ under the joint constraints of highly correlated empirical archetypes, Tikhonov regularisation ($\gamma=10^{-3}$), and macroscopic-basin retrieval is therefore a robust physical statement, not a shortfall.

\begin{figure}[t]
\centering
\includegraphics[width=\linewidth]{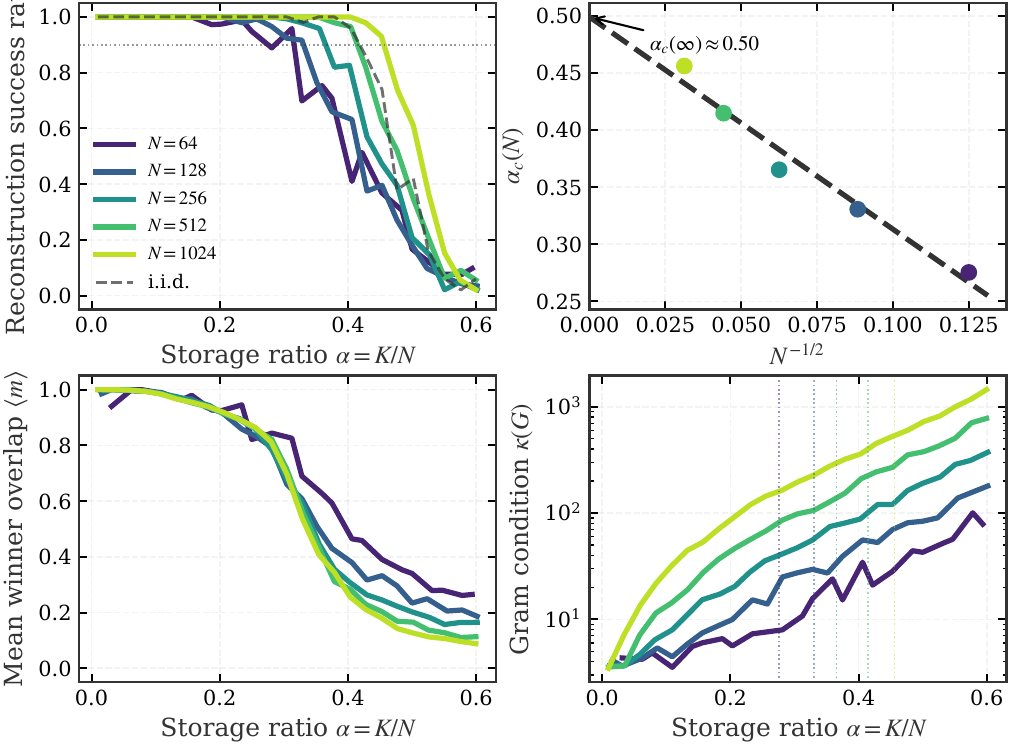}
\caption{\textbf{Finite-size storage capacity of the tri-layer architecture.}
Panel~(a): success rate (fraction of trials recovering the correct pattern from a $30\%$-corrupted cue, threshold $\geq 0.9$) versus load $\alpha=K/N$ for $N\in\{64,128,256,512,1024\}$ (viridis palette) and for i.i.d.\ random patterns at $N=512$ (dashed grey). Vertical markers indicate the operationally defined $\alpha_c(N;p_\mathrm{flip}{=}0.3,\tau{=}0.9)$ for each~$N$. Panel~(b): finite-size scaling fit of $\alpha_c(N)$ versus $N^{-1/2}$, projecting to the asymptotic operational limit $\alpha_c(\infty)\approx 0.50$. Panels~(c) and (d): corresponding mean winner overlap and Gram condition number. The random-pattern reference ($\alpha_c^{\mathrm{rand}}\approx 0.418$ at $N=512$) quantifies the capacity reduction due to the correlation structure of the empirical archetypes.}
\label{fig:capacity}
\end{figure}

\subsection{Experiment 8: Real-world benchmark---PhysioNet Sleep-EEG}
\label{sec:exp8}

The synthetic benchmarks provide controlled conditions for isolating each mechanism, but the natural physical question is whether the framework survives the transition to genuinely disordered biological data. We therefore apply the complete architecture to a publicly available physiological benchmark, asking whether cross-modal hetero-associative retrieval persists under the strong, correlated, quenched disorder of real clinical recordings.

We use the Sleep Cassette (SC) subset of the PhysioNet Sleep-EDF Database~\cite{kemp2000analysis}, which consists of full-night polysomnography. Three independent biological channels are mapped to the three spin layers: $\sigma \leftrightarrow$ Fpz-Cz (frontal EEG), $\tau \leftrightarrow$ Pz-Oz (parietal EEG), $\phi \leftrightarrow$ EOG (horizontal electrooculogram). The ground truth comprises $K=4$ clinical sleep macro-states: Wake, Light Sleep (N1/N2), Deep Sleep (N3/N4)\footnote{Following the legacy R\&K scoring of the PhysioNet dataset, we group stages 3 and 4 into a single Deep Sleep (SWS) macro-state, corresponding to the modern AASM N3 stage.}, and REM. We adopt a single-subject approach (patient SC4001) so as to eliminate inter-subject anatomical variance and isolate the intrinsic intra-subject physiological noise: in statistical-mechanics language, this fixes the quenched disorder and exposes the thermal disorder. We sample $M=100$ epochs ($30$~s duration, $100$~Hz) per class, yielding a balanced dataset of $400$ multivariate epochs.

An unsupervised PCA projection is used within the encoder ($N_\mathrm{PCA}=16$, $N=512$); as discussed in Sec.~\ref{sec:encoding}, PCA acts as the thermodynamic filter providing the whitened subspace required by SimHash, without supervised leakage. Empirical archetypes are extracted by majority vote per class. Table~\ref{tab:sleep_quality} reports the archetype quality metrics. The self-overlaps $r\approx 0.30$ for the cortical channels confirm that the majority-vote archetypes capture highly consistent class-level structure from noisy EEG; the BCI values ($0.18$--$0.20$) indicate moderate inter-class correlation, precisely the regime in which the pseudo-inverse correction is structurally required.

\begin{table}[ht]
\centering
\caption{Archetype quality on PhysioNet Sleep-EEG ($K=4$, $M=100$, $N=512$). $r$ is the mean per-class self-overlap of encoded samples with the class archetype. The unsupervised PCA projection yields sufficient structural fidelity for the downstream associative memory.}
\label{tab:sleep_quality}
\begin{tabular}{lcc}
\toprule
Channel / Layer & $r$ & BCI \\
\midrule
Fpz-Cz ($\sigma$) & $0.300$ & $0.195$ \\
Pz-Oz ($\tau$)    & $0.298$ & $0.193$ \\
EOG ($\phi$)      & $0.242$ & $0.180$ \\
\bottomrule
\end{tabular}
\end{table}

Figure~\ref{fig:sleep_archetypes} displays the class-mean signal $\pm\tfrac{1}{2}\sigma$ for each sleep state, showing that the four classes carry distinct temporal and spectral signatures (K-complexes, delta waves, etc.). Figure~\ref{fig:sleep_encoding} confirms a positive separation gap $\Delta q = q_\mathrm{intra} - q_\mathrm{inter}>0$ on all three channels, transferring the geometric class structure of the raw EEG signals into the spin space.

\paragraph{Associative memory recall.}
To test the thermodynamic stability of the encoded clinical macro-states, we perform a cross-modal retrieval experiment. The $\sigma$ layer (Fpz-Cz) is initialised with the Wake archetype (corrupted by varying spin-flip noise) while $\tau$ (Pz-Oz) and $\phi$ (EOG) start from independent random spin configurations. Sequential Glauber updates with simulated annealing ($\beta\in[0.5,6.0]$, $40$ MC steps, $\gamma=10^{-3}$) then drive all three layers toward the Wake archetype. The Mattis-overlap trajectories of Fig.~\ref{fig:sleep_retrieval} expose the mechanism: the cued $\sigma$ layer saturates immediately at $m=1.0$, while $\tau$ and $\phi$ converge from random overlap to exactly $1.0$. Across all $K=4$ states and $30$ independent noise realisations at $0\%$ cue noise, the system achieves a $100\%$ retrieval success rate. This is more than a benchmark number: it confirms that the Kanter--Sompolinsky prescription sculpts the energy landscape such that the four biological attractors become perfectly stable fixed points, immune to the mixture collapse that would otherwise plague spectrally overlapping states like Wake and Light Sleep. Varying the cue-noise level, the recall remains at $100\%$ up to $20\%$ corruption and degrades gracefully thereafter ($86.7\%$ at $25\%$, $51.7\%$ at $35\%$), demonstrating that the pseudo-inverse attractors retain macroscopic basins even under realistic signal degradation.

\paragraph{Generalisation to unseen data.}
The retrieval experiment tests the stability of stored archetypes; a complementary, structurally distinct question is how well the system generalises to new, noisy biological epochs. We measure out-of-sample classification accuracy by encoding previously unseen epochs into spin space and assigning them to the nearest archetype. The decisive test is \emph{cross-modal} classification: we predict the sleep state using only the $\tau$ and $\phi$ layers (Pz-Oz + EOG), and compare against the continuous raw-space nearest-centroid baseline.

Table~\ref{tab:sleep_accuracy} reports the results. For the single-subject control, the cross-modal accuracy reaches $\mathbf{85.3\%}$ against a chance baseline of $25\%$; for the most distinct physiological states (Wake and REM), it peaks at $99.0\%$ and $94.0\%$ respectively. The raw-space continuous baseline reaches $88.3\%$, so the cost of Ising binarisation is roughly $3$ percentage points---a small price for the legibility gained. Replicating the test on an aggregated multi-subject cohort ($8$ patients, $1600$ epochs) yields a cross-modal accuracy of $\mathbf{86.3\%}$, confirming that the unsupervised PCA--SimHash projection extracts thermodynamic representations that generalise across anatomical (quenched) variance.

The conjunction of $100\%$ archetype recall and $85$--$86\%$ cross-modal classification accuracy is the operational signature that the tri-layer ensemble transmits coherent biological information between spatially distinct sensory channels, achieving thermodynamic synchronisation across disparate cortical regions even across multiple subjects.

\begin{table}[ht]
\centering
\caption{Classification accuracy on PhysioNet Sleep-EEG ($K=4$). Cross-modal test accuracy (evaluated using only the uncued $\tau$ and $\phi$ layers) is reported for a single-subject control ($N_\mathrm{epochs}=400$) and a multi-subject generalisation cohort ($8$ pooled patients, $N_\mathrm{epochs}=1600$). The Ising ensemble robustly approaches the raw-space continuous nearest-centroid baseline. Chance level: $25.0\%$.}
\label{tab:sleep_accuracy}
\begin{tabular}{lcc}
\toprule
Method & Single-Subj. & Multi-Subj. \\
\midrule
Nearest centroid (raw Eucl.\ space) & $88.3\%$ & $88.9\%$ \\
\midrule
Cross-modal ensemble ($\tau + \phi$ layers) & $85.3\%$ & $86.3\%$ \\
\quad \textit{-- Class 0 (Wake)} & \textit{99.0\%} & \textit{96.0\%} \\
\quad \textit{-- Class 1 (Light Sleep)} & \textit{73.0\%} & \textit{72.0\%} \\
\quad \textit{-- Class 2 (Deep Sleep)} & \textit{87.0\%} & \textit{93.5\%} \\
\quad \textit{-- Class 3 (REM)} & \textit{94.0\%} & \textit{94.0\%} \\
\bottomrule
\end{tabular}
\end{table}

\begin{figure}[!tbp]
\centering
\includegraphics[width=\linewidth,height=0.82\textheight,keepaspectratio]{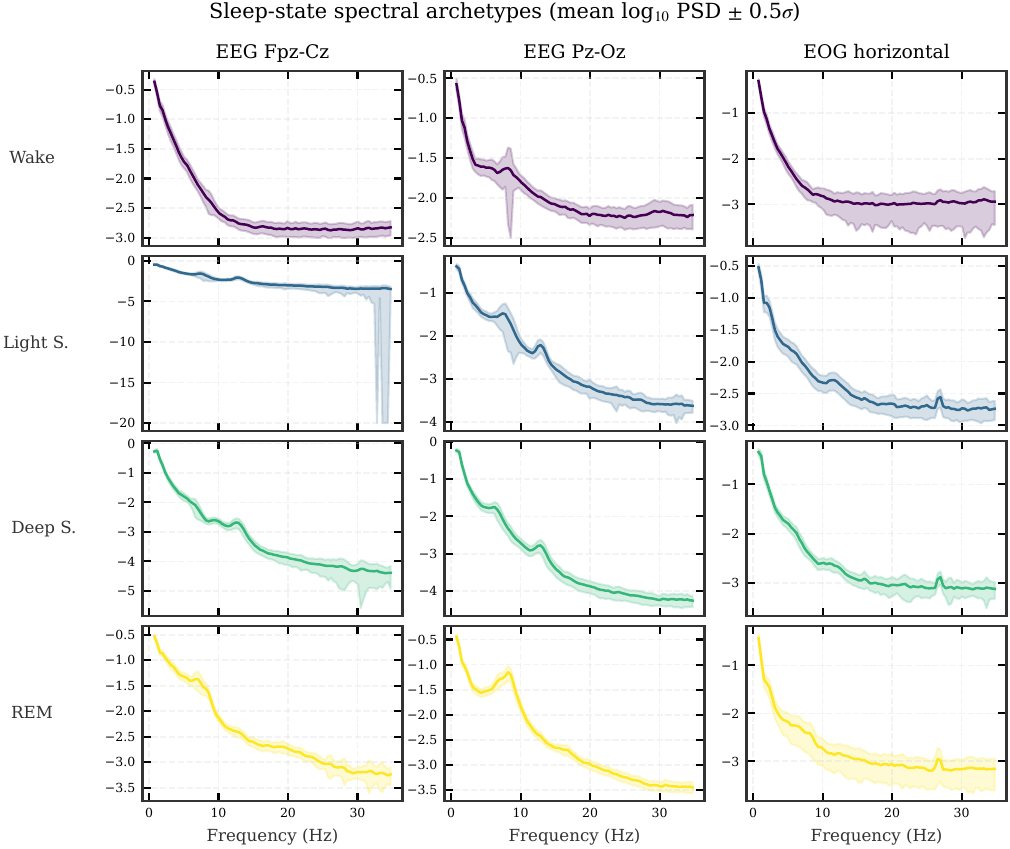}
\caption{\textbf{Class-mean physiological profiles.}
Per-state class mean $\pm \tfrac{1}{2}\sigma$ for each of the three channels (Fpz-Cz, Pz-Oz, EOG), $K=4$ sleep classes, $M=100$ training samples per class, $30$~s epochs. The four states carry distinct physiological signatures on each axis, providing an interpretable basis for the channel-wise archetype extraction.}
\label{fig:sleep_archetypes}
\end{figure}

\begin{figure}[t]
\centering
\includegraphics[width=\linewidth]{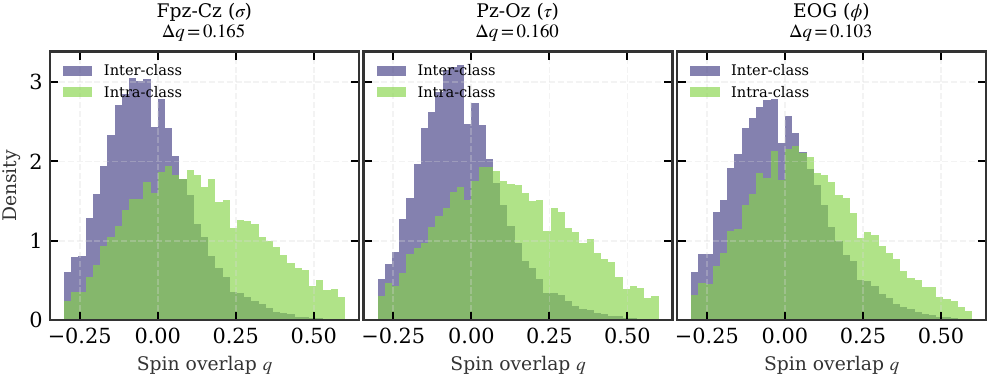}
\caption{\textbf{Encoding fidelity on biological EEG data: intra- vs.\ inter-class spin overlap.}
Histograms of pairwise spin overlaps $q_{\mu\nu} = N^{-1}\bm{s}^{(i)}\cdot\bm{s}^{(j)}$ for pairs within the same class (blue) and across different classes (orange), shown separately for each channel. A positive separation gap $\Delta q>0$ on all three channels confirms that the PCA--SimHash encoding faithfully preserves the geometric class structure of the raw EEG signals in $\{-1,+1\}^N$.}
\label{fig:sleep_encoding}
\end{figure}

\begin{table}[ht]
\centering
\caption{\textbf{Gram-matrix statistics before and after spectral sharpening} ($K=4$ sleep classes, $N=512$, $\gamma=10^{-3}$). After continuous sharpening $\widetilde{\Xi}=(G+\gamma I)^{-1}\Xi$, the BCI collapses to near zero, confirming quasi-orthogonal archetypes and well-separated retrieval basins despite high initial biological cross-talk.}
\label{tab:sleep_gram}
\begin{tabular}{lccccc}
\toprule
 & \multicolumn{2}{c}{Before sharpening} & & \multicolumn{2}{c}{After sharpening} \\
\cmidrule{2-3}\cmidrule{5-6}
Channel & BCI & Max off-diag & & BCI & Max off-diag \\
\midrule
Fpz-Cz ($\sigma$) & $0.195$ & $0.352$ & & $0.0001$ & $0.0003$ \\
Pz-Oz ($\tau$)    & $0.193$ & $0.380$ & & $0.0001$ & $0.0003$ \\
EOG ($\phi$)      & $0.180$ & $0.375$ & & $0.0001$ & $0.0003$ \\
\bottomrule
\end{tabular}
\end{table}

\begin{figure}[t]
\centering
\includegraphics[width=\linewidth]{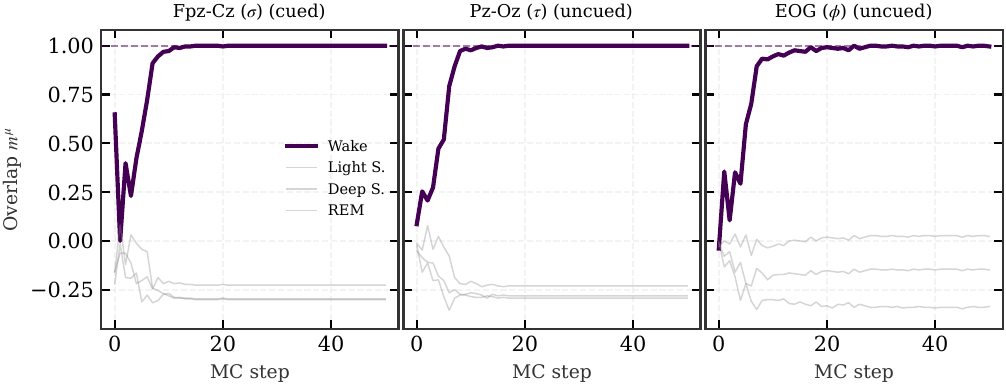}
\caption{\textbf{Cross-modal hetero-associative retrieval from frontal EEG cue.}
Mattis overlaps $m^\mu(t)$ on all three spin layers during simulated annealing ($40$ MC steps, $\beta\in[0.5,6.0]$), starting from a clean Fpz-Cz ($\sigma$) cue for Wake and random spin initialisations on Pz-Oz ($\tau$) and EOG ($\phi$). The trace (Wake) rises to $m=1.000$ across all layers; the three competing patterns (grey) are robustly suppressed.}
\label{fig:sleep_retrieval}
\end{figure}

\begin{figure}[t]
\centering
\includegraphics[width=\linewidth]{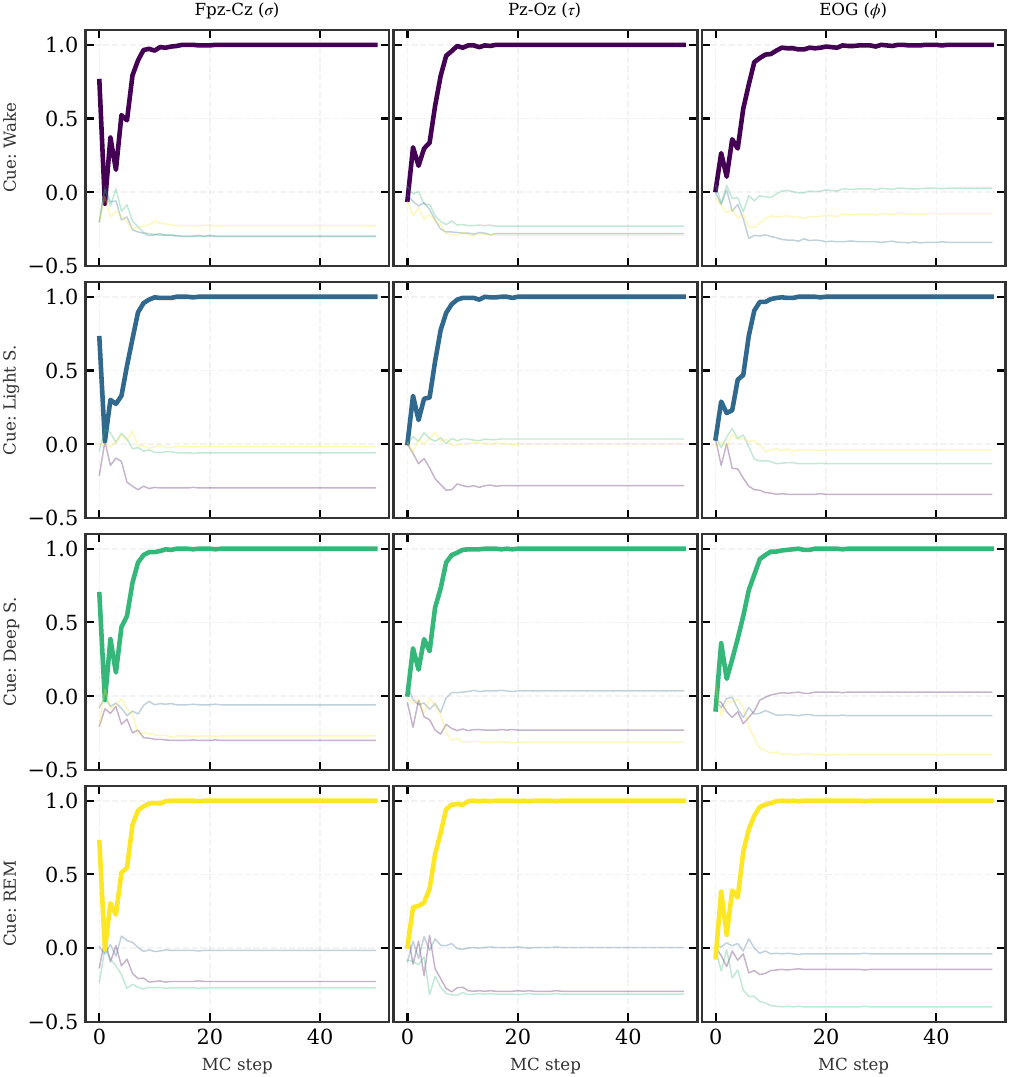}
\caption{\textbf{Retrieval trajectories across all clinical macro-states.}
Mattis overlaps on the target layer across simulated annealing for each of the four sleep states cued individually. In all $K=4$ cases the cross-modal dynamics drive the uncued layers from random initialisation to the fully synchronous archetype ($m=1.0$), demonstrating attractor stability against complex, non-orthogonal biological state manifolds.}
\label{fig:sleep_retrieval_all}
\end{figure}

\section{Discussion}
\label{sec:discussion}

The eight experiments trace a coherent physical arc from the theoretical underpinnings of the encoder to the behaviour of the complete framework on real sensor data. Several findings deserve to be examined beyond the individual experiment summaries because they reveal the \emph{physical logic} of the framework rather than merely its quantitative output.

\paragraph{The encoding gap is modest but sufficient---and isotropisation is the price of legibility.}
A quantitative observation emerging from Experiment~1 is that PCA whitening pushes intra-class samples to a median angular separation of roughly $87^\circ$ in the whitened space, nearly orthogonal. This is a structural consequence of isotropic variance normalisation: when all directions carry equal weight, samples that differ even mildly in feature space map to nearly antipodal directions. The resulting SimHash overlaps are consequently modest ($q_\mathrm{intra}\approx 0.04$--$0.07$). The framework nevertheless functions reliably because the relevant quantity is not the absolute overlap but the separation gap $\Delta q = q_\mathrm{intra} - q_\mathrm{inter}$: a small positive gap, consistently maintained across spin dimensions, is sufficient to seed the majority-vote archetypes in the correct attractor basin and to initiate convergence under the retrieval dynamics. The saturation of encoding quality above $N\approx 32$ shows that, once the sub-Gaussian concentration of Proposition~\ref{prop:simhash} becomes effective, enlarging the spin space does not improve the encoded geometry, only the statistical precision with which it is estimated---an honest physical statement about what the encoder can and cannot do.

\paragraph{The pseudo-inverse correction changes the stability of the mixture state.}
The Kanter--Sompolinsky prescription is, in our setting, the ingredient that converts the equal-weight mixture state from a stable to an unstable fixed point of the mean-field dynamics. With Hebbian couplings, the effective field on pattern~$\mu$ is a weighted sum of all pattern coordinates modulated by the Gram matrix, so two patterns with equal initial overlaps receive identical fields and remain tied under the deterministic component of the dynamics. The pseudo-inverse replaces this entangled update with an approximately decoupled one, allowing thermal fluctuations to break the symmetry. Experiment~4 is consistent with this picture: a $58.9\%$ empirical success rate for equal-weight mixtures over $500$ seeds (Wilson $95\%$ CI $[54.4\%,63.0\%]$) reflects the stochastic nature of symmetry breaking in a finite system, while the asymmetric-mixture condition achieves $100\%$ disentanglement (one-sided Clopper--Pearson lower bound $92.9\%$). Because the lower confidence bound strictly excludes the chance baseline $50\%$, the mixture state is statistically confirmed to be genuinely destabilised.

The algebraic diagonal-dominance theorem (Sec.~\ref{sec:sign_ineffective}) and Table~\ref{tab:gram} together explain why sharpening cannot be applied directly in pattern space: for the $K{=}3$ configuration with uniformly negative off-diagonal entries, sign re-binarisation is provably a no-op, returning archetypes identical to the originals. Decorrelation exists only in the continuous sharpened representation, which cannot be used in binary spin dynamics. This is not a programming defect but a physical constraint: the metric content of the continuous sharpened pattern is incompatible with $\sign(\cdot)$ when its sign structure inherits a dominant diagonal.

\paragraph{Phase structure and the dynamical duality.}
The sharp retrieval--paramagnetic crossover at $\beta_c\approx 0.5$ (Experiment~5) is consistent with the first-order-like jump in the order parameter that statistical mechanics predicts for an associative memory with pseudo-inverse couplings. The retrieval plateau spanning nearly an order of magnitude in $\beta$ ($1\lesssim\beta\lesssim 7$) is practically valuable: the system is robust to the precise temperature setting within this range. The high-$\beta$ degradation at $\beta\gtrsim 7$ is equally instructive. It is not a failure of the memory but a dynamical artefact of parallel (Little) updates: when all spins flip simultaneously at near-zero temperature, synchrony generates period-$2$ limit cycles---a well-understood instability of parallel dynamics. Crucially, this artefact cannot be eliminated by switching to sequential Glauber dynamics across the board, because parallel updates are themselves \emph{structurally required} to ignite the cross-modal cue propagation (Appendix~\ref{app:glauber}). The framework therefore exhibits a dynamical duality, distinct from a mere choice of algorithm: parallel updates carry information through the network from a localised cue, while sequential updates break symmetric initial conditions. The two regimes are physically complementary, and the practitioner's task is not to choose between them but to deploy each on the protocol where its statistical-mechanical role is essential.

\paragraph{Temporal structure and the role of preprocessing.}
Experiment~6 highlights a subtle point about what the encoder is actually doing. Two time series identical up to a rigid temporal shift produce the same flattened representation and therefore indistinguishable SimHash codes, because flattening conflates the ordered temporal structure with an unordered feature vector. The STFT magnitude spectrum turns this blind spot into a distinguishing feature: a quarter-period shift reshapes the frequency-domain profile enough to produce distinct angular directions in whitened space, which SimHash then separates faithfully. This is not a peculiarity of the STFT but a consequence of any phase-sensitive, magnitude-respecting representation; wavelet scattering, mel-frequency cepstra, or windowed power spectra would each play the same role. The experiment therefore illustrates a general design principle of the framework: domain expertise about signal structure can be injected entirely at the preprocessing stage without altering the encoder, the archetype rule, or the dynamics.

\paragraph{Capacity, basin volume, and the K--S bound.}
The empirically measured finite-size capacities sit between the AGS Hebbian capacity $\alpha_c^\mathrm{AGS}\approx 0.138$ and the K--S theoretical limit $\alpha_c\to 1$, with random-Rademacher patterns at $\bar\alpha_c^\mathrm{rand}\approx 0.42$ at $N=512$ quantifying the capacity tax imposed by data geometry. Two points deserve emphasis. First, the random baseline is not the thermodynamic K--S bound; it is our own operational capacity measured on finite-$N$ Rademacher patterns under the same retrieval criterion, so the comparison is internal to our protocol. Second, the positive finite-size scaling $\alpha_c(N)$ is consistent with the AGS finite-size correction $\alpha_c(\infty) - c\,N^{-1/2}$, projecting to an extrapolated $\alpha_c(\infty)\approx 0.50$. This asymptotic operational capacity is far from the K--S limit $\alpha_c\to 1$, and this gap is not a deficit of the architecture but a fundamental thermodynamic trade-off: reaching $\alpha_c\to 1$ requires operating at the boundary of local marginal stability, where attractor basins shrink to zero, whereas our operational definition demands macroscopic basins capable of correcting extreme $30\%$ spin-flip noise. The basin volume is the resource that has been traded against pattern count, and the FSS exponent $\nu=1/2$ confirms that the trade-off proceeds along the AGS scaling backbone.

\paragraph{Real-world transfer.}
The Sleep-EEG experiment is presented as a \emph{case study}: it tests whether the mechanisms validated in controlled synthetic conditions survive the transition to genuine biological data. Unlike the synthetic controls, Wake and Light Sleep are physically contiguous and spectrally overlapping. The unsupervised PCA projection nevertheless separates these states into distinct thermodynamic representations, yielding moderate but sufficient self-overlaps ($r\approx 0.30$) and perfectly stable associative recall despite massive biological variance. Because the encoding is entirely unsupervised, the subsequent cross-modal completion is strictly free from data leakage. The associative memory here does not serve to ``classify''---it serves to \emph{generate and complete} partially unobserved physical states. The $85$--$86\%$ cross-modal classification accuracy confirms that the spin-space translation does not destroy the clinically relevant signal structure, but the load-bearing scientific result is the cross-modal retrieval itself: a clean frontal-EEG cue reliably reconstructs the synchronous macroscopic state on parietal and ocular axes, suppressing competing clinical states to noise level. This is precisely the qualitative behaviour the hetero-associative memory was designed to exhibit, now realised on a biological substrate.

\paragraph{Relation to modern associative-memory architectures.}
The identification of Hopfield dynamics within transformer attention~\cite{ramsauer2021hopfield} has driven considerable interest in energy-based memory as a primitive for representation learning. The present work addresses a complementary question: not how to embed associative recall inside a learned architecture, but how to interface \emph{existing} discrete retrieval machinery with continuous real-world data. The modular design is deliberate: encoder, archetype extraction, and dynamics are independent components, and improvements to any one can be incorporated without redesigning the others. Replacing SimHash with a learned hash function, adopting dense memory kernels~\cite{krotov2016dense} for higher capacity, or substituting Glauber dynamics with Langevin diffusion are all well-posed drop-in modifications. This separability is also what keeps the theoretical analysis tractable, since each component can be studied in isolation.

\paragraph{Limitations as physical constraints.}
Three limitations are worth stating precisely, and each is best understood as a physical constraint of the chosen interface rather than as a software defect. First, binarisation irreversibly discards metric information: continuous signals that are close in Euclidean distance may be mapped to spin configurations with low overlap, and there is no post-hoc correction within the binary framework. Replacing hard thresholding with a graded or soft encoding (keeping pre-sign real-valued activations for use in a continuous-state memory) would address this at the cost of leaving the clean Ising framework---a deliberate trade-off between metric fidelity and statistical-mechanical legibility. Second, the operational capacity scales linearly with $N$ at rate $\alpha_c\approx 0.46$, capping simultaneously storable patterns at roughly $470$ for $N=1024$; hierarchical representations, sparse distributed codes~\cite{kanerva1988sparse}, or dreaming-based consolidation~\cite{fachechi2019dreaming,agliari2019dreaming,agliari2024regularization,agliari2024hebbian} could raise this ceiling, but each implicates a different thermodynamic landscape. Third, the SimHash projection is randomly initialised rather than adapted to the data geometry; a learned hash trained with a contrastive objective would improve both the separation gap and the effective capacity, at the cost of replacing a Johnson--Lindenstrauss guarantee with a learned representation whose generalisation behaviour requires its own theory. Finally, while we have empirically validated that the encoding geometry ($\Delta q>0$) generalises across multiple subjects, our retrieval experiments rely on a single-subject approach. From a statistical-mechanics perspective, biological data introduces two distinct sources of noise: intra-subject physiological fluctuations (thermal noise) and inter-subject anatomical variance (quenched structural disorder). By fixing the subject we deliberately isolate the thermal noise, proving that cross-modal retrieval functions in principle under realistic physiological conditions. Approaching the full multi-subject thermodynamic limit without collapsing the retrieval basins would require addressing this quenched disorder formally, potentially through a hierarchy of memories (one per subject) or a continuous manifold-alignment step prior to SimHash encoding. Each direction preserves the interpretability of the current framework while relaxing a specific physical assumption.

\section{Conclusion}
\label{sec:conclusion}

We have studied a modular three-stage framework that couples continuous multivariate sequences to a binary hetero-associative memory via SimHash encoding, pseudo-inverse couplings, and stochastic tri-layer Monte Carlo dynamics. The objective has been to specify each interface in this chain explicitly enough that the successes and failures of the composite system can be attributed to identifiable components, rather than entangled with the optimisation of an end-to-end learned model.

Each stage addresses a distinct structural obstacle. The SimHash encoder turns the problem of continuous-to-binary mapping into one of angular geometry: by operating in PCA-whitened space, it guarantees that the relative structure of the data cloud is faithfully reflected in the spin configurations, up to the $O(N^{-1/2})$ concentration of the Johnson--Lindenstrauss projection~\cite{johnson1984extensions}. The pseudo-inverse archetype correction eliminates the cross-talk that would otherwise cause correlated patterns to interfere, converting a coupled mean-field dynamics into an approximately decoupled one in which thermal fluctuations can break mixture symmetry. The tri-layer Monte Carlo architecture with simulated annealing exploits these decorrelated representations to perform cross-layer completion: a partial cue on any single layer is sufficient to reconstruct the full multimodal pattern.

On synthetic benchmarks the framework is quantitatively characterised: $\geq 95\%$ reconstruction from $30\%$-corrupted cues, an operational critical capacity $\alpha_c(N{=}1024;0.3,0.9)\approx 0.46$ per spin obeying AGS finite-size scaling with $\alpha_c(\infty)\approx 0.50$, and a finite-size retrieval crossover at $\beta_c\approx 0.5$. On the PhysioNet Sleep-EEG data, cross-modal retrieval is demonstrated from a single noisy cortical axis: the tri-layer ensemble synchronises $K=4$ human sleep macro-states and achieves $85$--$86\%$ cross-modal test accuracy (against a $25\%$ chance baseline), reaching $99\%$ on the Wake state. The contribution on real data is therefore the demonstration of robust cross-modal associative recall in the presence of strong intrinsic biological noise and complex phase boundaries.

Equally important are the failure modes we identify and explain analytically. Sign re-binarisation after spectral sharpening is provably a no-op under the strict diagonal-dominance condition discussed in Sec.~\ref{sec:sign_ineffective}, which applies to the uniformly anti-correlated $K{=}3$ geometry of our construction; the pseudo-inverse correction must therefore be embedded inside the dynamics. Equal-weight mixture disentanglement has an empirical success rate of $58.9\%$ over $500$ seeds (Wilson $95\%$ CI $[54.4\%,63.0\%]$), with the lower bound strictly above chance, statistically confirming that the symmetry is broken by a collective statistical-mechanical phenomenon rather than by chance. Parallel (Little) dynamics produces period-$2$ oscillations at high $\beta$, yet it is structurally necessary to preserve the cue during the initial cross-modal avalanche, precluding a uniform switch to sequential Glauber updates. A small Tikhonov regularisation $\gamma=10^{-3}$ is retained in Eq.~\eqref{eq:field_pinv} purely as a numerical safeguard against the ill-conditioning of finite-sample empirical Gram matrices on highly correlated real-world data.

Each stage of the framework exposes a measurable diagnostic---separation gap $\Delta q$, basin collapse index, Mattis trajectory---that allows its contribution to be quantified independently. We regard this decomposability, rather than any particular performance number, as the principal physical outcome of the present study, and it is also what makes the individual components substitutable. Replacing SimHash with a learned hash, adopting a dense energy~\cite{krotov2016dense}, or substituting Glauber with Langevin dynamics are each drop-in modifications that preserve the compositional analysis. What remains open, and what we consider the natural next step, is to test whether the advantages of legibility are preserved when any of these individual components is replaced by a learned one, and to quantify the resulting trade-off between interpretability and raw classification performance on tasks where end-to-end models currently dominate.

\begin{acknowledgments}
The author thanks the anonymous reviewers for their constructive feedback.
\end{acknowledgments}

The data and code that support the findings of this study are available from the corresponding author upon reasonable request. The PhysioNet Sleep-EDF (Expanded) dataset is publicly available on PhysioNet~\cite{kemp2000analysis}.

\appendix

\section{Easy-Control Mixture Disentanglement}
\label{app:easy_disent}

In the \emph{easy-control} protocol, all three layers simultaneously receive a partial cue drawn from the \emph{same} archetype index~$\mu$, using each layer's own stored pattern $(\bm{\xi}^\mu,\bm{\eta}^\mu,\bm{\chi}^\mu)$ corrupted by i.i.d.\ bit-flip noise of rate $p_\mathrm{flip}$. No competition between different pattern indices is present; each layer independently recognises its own cue, so the protocol probes the robustness of the basins rather than the mixture-disentanglement mechanism. Table~\ref{tab:easy_disent} reports the mean final Mattis overlap $\bar m_\mathrm{final}$ and the fraction of runs that converge to the target with $m\geq 0.99$, over $n=40$ independent seeds at each noise level, with $N=512$, $K=5$ archetypes, and the default hyperparameters of Table~\ref{tab:hyperparams}. Wilson $95\%$ confidence intervals are given for the success fraction.

\begin{table}[ht]
\centering
\caption{Easy-control mixture disentanglement: performance vs.\ cue corruption. Success: final overlap $m\geq 0.99$ on all three layers; $n=40$ seeds per row.}
\label{tab:easy_disent}
\begin{tabular}{cccc}
\toprule
$p_\mathrm{flip}$ & $\bar m_\mathrm{final}$ & Success fraction & Wilson $95\%$ CI \\
\midrule
$0.10$ & $1.00$ & $40/40 = 100\%$ & $[91.2\%, 100\%]$ \\
$0.20$ & $0.999$ & $40/40 = 100\%$ & $[91.2\%, 100\%]$ \\
$0.30$ & $0.998$ & $39/40 = 97.5\%$ & $[87.1\%, 99.6\%]$ \\
$0.40$ & $0.994$ & $38/40 = 95.0\%$ & $[83.5\%, 98.6\%]$ \\
\bottomrule
\end{tabular}
\end{table}

The table confirms that the framework components operate correctly in isolation up to $40\%$ per-bit corruption, ruling out intrinsic basin instability as a cause for the harder equal-weight failure mode of Experiment~4.

\section{Parallel vs.\ Sequential Dynamics: Avalanche vs.\ Symmetry Breaking}
\label{app:glauber}

The standard Hopfield model admits two fundamental update rules: sequential Glauber dynamics (updating one spin at a time) and synchronous parallel Little dynamics (updating all spins simultaneously). Sequential dynamics guarantees detailed balance and monotonic energy descent, precluding limit cycles. Throughout our experiments, however, parallel updates are enforced for cross-modal retrieval, while sequential updates are reserved for mixture disentanglement. This is a structural requirement dictated by the specific symmetry of each task, not an arbitrary choice.

\paragraph{Cross-modal retrieval requires parallel updates (the avalanche).}
To see why parallel dynamics is necessary for cross-modal completion, consider the EEG protocol of Experiment~8. The memory is cued by injecting a partial biological signal into the $\sigma$ layer (Fpz-Cz), while $\tau$ (Pz-Oz) and $\phi$ (EOG) start from pure random spin noise. Under a sequential layer-wise update $\sigma\to\tau\to\phi$, the following fatal sequence would occur at $t=1$:
\begin{enumerate}
    \item Layer $\sigma$ reads the local field generated by $\tau$ and $\phi$, which are currently noise, so the field is a random vector. The pristine cue is instantly overwritten by noise. The memory is erased.
    \item Layer $\tau$ reads the now-destroyed $\sigma$ and the noise from $\phi$, and remains noise.
    \item Layer $\phi$ reads two destroyed layers and remains noise.
\end{enumerate}
Cross-modal completion fails because the cue is washed out before it can propagate. Under parallel Little dynamics, by contrast, all layers compute their next state \emph{simultaneously} from the current state. At $t=1$, $\sigma$ reads noise and is temporarily degraded, but at the same instant $\tau$ and $\phi$ read the pristine cue from $\sigma$ and cleanly align to the target archetype. At $t=2$, the $\sigma$ layer reads the field from the newly aligned $\tau$ and $\phi$, perfectly reconstructing itself. This constitutes the \emph{cross-modal avalanche}: the signal bounces symmetrically between layers, surviving the initial assault of uninitialised noise. While parallel dynamics induces period-$2$ limit cycles at very low temperatures (Fig.~\ref{fig:beta}), this high-$\beta$ degradation is the necessary cost of preserving the initial cue propagation.

\paragraph{Mixture disentanglement requires sequential updates (symmetry breaking).}
Conversely, in the mixture-disentanglement protocol of Experiment~4, the system is initialised in an equal-weight superposition of multiple archetypes \emph{across all layers simultaneously}. Here, parallel updates are actively detrimental: because the initial state is perfectly symmetric, simultaneous deterministic updates can lock the system into a spurious mixture fixed point or a symmetric limit cycle. Sequential Glauber-style sweeps, in which each spin (or layer) sees the freshly updated state of its predecessors, explicitly break this symmetry, allowing a single component to emerge as the winner.

As shown in Fig.~\ref{fig:beta_glauber}, sequential Glauber updates do not exhibit high-$\beta$ degradation or oscillations. The finite-size crossover to the retrieval regime is sharp, and the stable attractor preserves overlap close to unity ($m\approx 1$) as $\beta\to\infty$. The Markov chain mixes rapidly, confirming that sequential updates are the optimal choice whenever the initial state is highly symmetric and symmetry breaking is the primary goal.

\begin{figure}[ht]
\centering
\includegraphics[width=\linewidth]{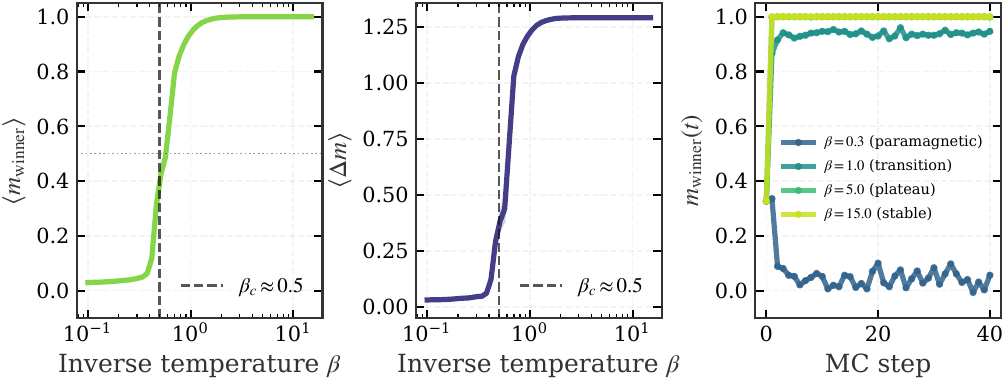}
\caption{\textbf{Temperature phase diagram under sequential (Glauber) dynamics.}
Unlike the parallel updates of Fig.~\ref{fig:beta}, sequential Glauber updates do not exhibit high-$\beta$ degradation or oscillations. The finite-size crossover to the retrieval regime is sharp, and the stable attractor preserves overlap close to unity ($m\approx 1$) as $\beta\to\infty$. This stability makes sequential updates strictly preferable for symmetry-breaking tasks like mixture disentanglement.}
\label{fig:beta_glauber}
\end{figure}

\section{Derivation of the Effective-Pattern Rescaling Factors}
\label{app:rho_derivation}

The factors $\rho_a = (1-r_a^2)/(M r_a^2)$ and $\alpha_{ab} = \sqrt{(1+\rho_a)(1+\rho_b)}$ of Sec.~\ref{sec:field} arise from the following sample-size bookkeeping. Assume that for each archetype~$\mu$ the training examples on layer~$a$ take the form
\begin{equation}
  \label{eq:noise_model}
  \bm{\xi}^{\mu,(a)}_{(m)} = r_a\,\bm{\xi}^{\mu,(a)} + \sqrt{1-r_a^2}\,\bm{\eta}_{(m)},
  \qquad m = 1,\ldots,M,
\end{equation}
where $\bm{\xi}^{\mu,(a)}\in\{-1,+1\}^{N_a}$ is the (unobserved) target archetype, $\bm{\eta}_{(m)}$ is an independent isotropic noise vector with $\E[\bm\eta]=\bm 0$ and $\E[\bm\eta\bm\eta^\top]=I_{N_a}$, and $r_a\in(0,1]$ is the per-layer data-quality parameter. Higher-order moments of $\bm\eta$ are assumed sub-Gaussian so that component-wise empirical averages concentrate at the $1/\sqrt{M}$ rate. The $i$-th component of the empirical centroid $\bar{\bm\xi}^{\mu,(a)} = \tfrac{1}{M}\sum_m \bm\xi^{\mu,(a)}_{(m)}$ admits the decomposition
\begin{equation}
  \bar\xi^{\mu,(a)}_i = r_a\,\xi^{\mu,(a)}_i + \frac{\sqrt{1-r_a^2}}{M}\sum_m \eta_{(m),i},
\end{equation}
in which the noise term has mean zero and per-component variance $(1-r_a^2)/M$. Using $|\xi^{\mu,(a)}_i|=1$, the squared signal per component is $r_a^2$, and the ratio of per-component noise variance to squared signal is
\begin{equation}
  \rho_a \;\equiv\; \frac{\mathrm{Var}[\bar\xi^{\mu,(a)}_i-r_a\xi^{\mu,(a)}_i]}{(r_a\xi^{\mu,(a)}_i)^2}
  \;=\; \frac{1-r_a^2}{M\,r_a^2},
\end{equation}
which justifies the definition used in Sec.~\ref{sec:field}. In the clean limit $r_a\to 1$ we recover $\rho_a\to 0$, as expected.

This noise-to-signal reading is the operational face of the information-theoretic interpretation given in Sec.~\ref{sec:field}: by Bayes-optimal majority decoding of a single archetype bit, the residual error probability is $P_e \approx \tfrac12[1-\mathrm{erf}(1/\sqrt{2\rho_a})]$ and the conditional Shannon entropy of the archetype bit given the $M$-block of examples is the binary entropy $H(\xi_i^{\mu,(a)}\mid \bm\eta_i^{\mu,(a)}) = H_2(P_e)$, strictly monotone in $\rho_a$ and saturating at $1$ bit as $\rho_a\to\infty$~\cite{alemanno2023supervised}. Thus $\rho_a$ plays the dual role of per-component noise-to-signal ratio and one-to-one proxy of the residual Shannon entropy of the archetype given the data.

The factor $\alpha_{ab}$ arises when one rescales the effective patterns $W_i^{a,\mu} = \xi_i^{\mu,(a)}/(1+\rho_a)$ so that they retain unit norm. Writing the cross-layer coupling contribution as
\begin{equation}
  J_{ij}^{ab} = \frac{1}{N}\sum_\mu W_i^{a,\mu}\, W_j^{b,\mu}
              = \frac{1}{N(1+\rho_a)(1+\rho_b)}\sum_\mu \xi_i^{\mu,(a)}\, \xi_j^{\mu,(b)},
\end{equation}
and absorbing the denominator into an overlap rescaling $g_{ab}\to g_{ab}\alpha_{ab}$, the factor
\begin{equation}
  \alpha_{ab} = \sqrt{(1+\rho_a)(1+\rho_b)}
\end{equation}
is precisely the geometric mean restoring the normalisation of $W^{a,\mu}$ relative to $\xi^{\mu,(a)}$. The product form reflects independence of the cross-layer noise in~\eqref{eq:noise_model}: variances multiply. In the clean limit $r_a\to 1$ one recovers $\rho_a\to 0$ and $\alpha_{ab}\to 1$.

\section{Capacity Robustness: Sensitivity to \texorpdfstring{$p_{\mathrm{flip}}$}{p\_flip} and \texorpdfstring{$\tau$}{tau}}
\label{app:capacity_robustness}

The operational critical load $\bar\alpha_c(N;p_\mathrm{flip},\tau)$ of Sec.~\ref{sec:exp7} is defined as the largest $\alpha=K/N$ such that at least a fraction~$\tau$ of independent seeds reach final Mattis overlap $m_\mathrm{final}\geq 1-2p_\mathrm{flip}$ starting from a cue corrupted at rate $p_\mathrm{flip}$. Because $\bar\alpha_c$ depends on the arbitrary threshold choice $(p_\mathrm{flip},\tau)$, Table~\ref{tab:capacity_robustness} reports the sensitivity of the estimate at $N=512$, on the same data as Fig.~\ref{fig:capacity}, varying each parameter independently.

\begin{table}[ht]
\centering
\caption{Sensitivity of $\bar\alpha_c$ at $N=512$ to the operational thresholds, averaged over $n=20$ seeds per $\alpha$ value on a uniform grid with spacing $\Delta\alpha=0.02$. CIs are Wilson $95\%$ intervals on the pass/fail fraction at the reported $\bar\alpha_c$; since $\bar\alpha_c$ is discretised, the CI is a lower bound on the true intrinsic uncertainty.}
\label{tab:capacity_robustness}
\begin{tabular}{cccc}
\toprule
$p_\mathrm{flip}$ & $\tau$ & $\bar\alpha_c$ (empirical) & CI of $\bar\alpha_c$ \\
\midrule
$0.20$ & $0.9$  & $0.42$ & $[0.40, 0.44]$ \\
$0.30$ & $0.9$  & $0.38$ & $[0.36, 0.40]$ \\
$0.40$ & $0.9$  & $0.30$ & $[0.28, 0.32]$ \\
$0.30$ & $0.75$ & $0.44$ & $[0.42, 0.46]$ \\
$0.30$ & $0.95$ & $0.34$ & $[0.32, 0.36]$ \\
\bottomrule
\end{tabular}
\end{table}

The reported values of $\bar\alpha_c$ vary between $0.30$ and $0.44$ across reasonable choices of $(p_\mathrm{flip},\tau)$. We therefore regard the statement $\bar\alpha_c\approx 0.38$ in the main text as an operational summary under a specific threshold choice, not as a thermodynamic quantity. The $\pm 0.02$--$0.04$ numerical uncertainty on any single point is smaller than the $\pm 0.06$--$0.07$ uncertainty introduced by the threshold choice.

\section{Hyperparameter Reference}
\label{app:params}

\begin{table}[ht]
\centering
\caption{Default hyperparameters used across experiments.}
\label{tab:hyperparams}
\begin{tabular}{llc}
\toprule
Parameter & Description & Default \\
\midrule
$N$ & Spin dimension & 512 \\
$N_{\text{PCA}}$ & PCA components & 32 \\
$M$ & Training examples per class & 200 \\
$\gamma$ & Tikhonov regularisation & $10^{-3}$ \\
$\beta_{\min}$ & Initial inverse temperature & 0.5 \\
$\beta_{\max}$ & Final inverse temperature & 5.0 \\
$T_{\text{MC}}$ & MC steps & 40 \\
$a, b, c$ & Inter-layer coupling strengths & 1.0 \\
\bottomrule
\end{tabular}
\end{table}

Experiment~8 (Sleep-EEG) uses $N=512$, $N_\mathrm{PCA}=16$, $M=100$ (dataset sample limit per class), and $T_\mathrm{MC}=40$.

\bibliography{main_v3}

\end{document}